\documentclass[11pt]{article}

\usepackage[utf8]{inputenc}
\usepackage[T1]{fontenc}
\usepackage{lmodern}
\usepackage{dsfont}

\usepackage[backend=biber,
            sorting=none]
          {biblatex}
\bibliography{refs}

\usepackage{csquotes}

\usepackage[colorlinks=true,
            citecolor=blue,
            urlcolor=blue,
            linkcolor=blue]
           {hyperref}

\usepackage{amsmath}
\usepackage{amsfonts}
\usepackage{amssymb}
\usepackage{amsthm}

\usepackage{mathtools}
\DeclarePairedDelimiter{\abs}{\lvert}{\rvert}
\DeclarePairedDelimiter{\ket}{\lvert}{\rangle}

\DeclarePairedDelimiterX{\braket}[2]{\langle}{\rangle}{#1\delimsize\vert #2}
\DeclarePairedDelimiterX{\ketbra}[2]{\lvert}{\rvert}{#1\delimsize\rangle\delimsize\langle #2}
\DeclarePairedDelimiterX{\expval}[3]{\langle}{\rangle}{#1\delimsize\vert #2\delimsize\vert #3}
\DeclarePairedDelimiter{\gen}{\langle}{\rangle}
\DeclarePairedDelimiter{\set}{\{}{\}}
\DeclarePairedDelimiterX{\setbuilder}[2]{\{}{\}}{#1\mathrel{\delimsize\vert} #2}

\newcommand{\diag}{\mathop{}\mathopen{}\mathrm{diag}}
\newcommand{\poly}{\mathop{}\mathopen{}\mathrm{poly}}

\newtheorem{theorem}{Theorem}
\newtheorem{proposition}{Proposition}
\newtheorem{lemma}{Lemma}

\newtheorem{corollary}{Corollary}
\newtheorem*{problem}{Problem}
\theoremstyle{definition}
\newtheorem{definition}{Definition}
\theoremstyle{remark}
\newtheorem{remark}{Remark}

\usepackage{booktabs}

\usepackage{authblk}

\begin{document}

\title{A Non-Commuting Stabilizer Formalism}
\author[1]{Xiaotong Ni}
\author[2]{Oliver Buerschaper}
\author[1]{Maarten Van den Nest}
\affil[1]
      {Max-Planck-Institut für Quantenoptik,
       Garching,
       Germany}
\affil[2]
      {Perimeter Institute for Theoretical Physics,
       31 Caroline Street North,
       Waterloo, Ontario,
       Canada, N2L\,2Y5}

\maketitle

\begin{abstract}
    We propose a non-commutative extension of the Pauli stabilizer formalism.
    The aim is to describe a class of many-body quantum states which is richer than the standard Pauli stabilizer states.
    In our framework, stabilizer operators are tensor products of single-qubit operators drawn from the group $\langle \alpha I, X,S\rangle$, where $\alpha=e^{i\pi/4}$ and $S=\operatorname{diag}(1,i)$.
    We provide techniques to efficiently compute various properties related to bipartite entanglement, expectation values of local observables, preparation by means of quantum circuits, parent Hamiltonians etc.
    We also highlight significant differences compared to the Pauli stabilizer formalism.
    In particular, we give examples of states in our formalism which cannot arise in the Pauli stabilizer formalism, such as topological models that support non-Abelian anyons.
\end{abstract}

\newcommand{\PauliGroup}{\mathcal{P}}
\newcommand{\PauliGroupMany}{\mathcal{P}_n}
\newcommand{\PauliSGroup}{\mathcal{P}^S}
\newcommand{\PauliSGroupMany}{\mathcal{P}_n^S}
\newcommand{\PermutationGroup}{\mathfrak{P}}

\newcommand{\Generators}{\mathcal{S}}
\newcommand{\Commutators}{\mathcal{C}_\mathcal{S}}
\newcommand{\Squares}{\mathcal{Q}_\mathcal{S}}
\newcommand{\MagicZ}{\mathcal{Z}}
\newcommand{\DiagonalSubgroup}{G_D}
\newcommand{\ZSubgroup}{G_Z}
\newcommand{\GeneratorsDiagonal}{\mathcal{D}}

\newcommand{\Identity}{I}
\newcommand{\ControlledZ}{\mathrm{C}Z}
\newcommand{\ControlledZLong}{\diag(1,1,1,-1)}
\newcommand{\ControlledControlledZ}{\mathrm{CC}Z}
\newcommand{\ControlledNOT}{\mathrm{CNOT}}
\newcommand{\ControlledS}{\mathrm{C}S}
\newcommand{\ControlledSLong}{\diag(1,1,1,\mathrm{i})}
\newcommand{\LogicalZ}{\bar{Z}}
\newcommand{\LogicalX}{\bar{X}}

\newcommand{\CliffordCircuit}{\mathcal{C}}
\newcommand{\TCSCCZCircuit}{\mathcal{U}}
\newcommand{\ControlledNOTCircuit}{\mathcal{C}_\ControlledNOT}

\newcommand{\ComplexityClass}[1]{\textsc{#1}}
\newcommand{\Problem}[1]{\textsc{#1}}

\newcommand{\Orbit}{\mathcal{O}}
\newcommand{\CovariantFunctions}{\mathcal{F}}
\newcommand{\PauliSVertexStabilizer}{\mathcal{A}}

\newcommand{\OrbitSpace}{V}
\newcommand{\InvariantSetDiagonal}{V_D}
\newcommand{\InvariantSpaceZ}{V_Z}

\newcommand{\CSSCode}[4]{\mathrm{CSS}(#1,#2;#3,#4)}
\newcommand{\ReedMullerCode}{\mathcal{L}}
\newcommand{\RegularXSCode}{\mathcal{H}_G}
\newcommand{\XSCode}{\mathcal{L}_G}

\newcommand{\IdentityMatrix}{\mathds{1}}

\section{Introduction}

Harnessing the properties of many-body entangled states is one of the central aims of quantum information theory.
An important obstacle in understanding many-particle systems is the exponential size of the Hilbert space i.e. exponentially many parameters in $N$ are needed to write down a general quantum state of $N$ particles.
One valid strategy to deal with this problem is to study subclasses of states that may be described with considerably less parameters, while maintaining a sufficiently rich structure to allow for nontrivial phenomena.
The Pauli stabilizer formalism (\textsc{PSF}) is one such class and it is a widely used tool throughout the development of quantum information~\cite{Gottesman1997Stabilizer}.
In the \textsc{PSF}, a quantum state is described in terms of a group of operators that leave the state invariant.
Such groups consist of Pauli operators and are called Pauli stabilizer groups.
An $n$-qubit Pauli operator is a tensor product $g=g^{(1)}\otimes\dots\otimes g^{(n)}$ where each $g^{(i)}$ belongs to the single-qubit Pauli group, i.e. the group generated by the Pauli matrices $X$ and $Z$ and the diagonal matrix~$\mathrm{i}\Identity$.
Since every stabilizer group is fully determined by a small set of generators, the \textsc{PSF} offers an efficient means to describe a subclass of quantum states and gain insight into their properties.
States of interest include the cluster states \cite{Raussendorf2003measurement}, \textsc{GHZ} states \cite{Greenberger1990bell} and the toric code \cite{Kitaev2003fault}; these are entangled states which appear in the contexts of e.g. measurement based quantum computation \cite{Raussendorf2003measurement} and topological phases.
 
Considering the importance of the \textsc{PSF}, it is natural to ask whether we can extend this framework and describe a larger class of states, while keeping as much as possible both a transparent mathematical description and computational efficiency.
In this paper, we provide a generalization of the \textsc{PSF}.
In our setting, we allow for stabilizer operators which are tensor product operators $g^{(1)}\otimes\dots\otimes g^{(n)}$  where each $g^{(i)}$ belongs to the group generated by the matrices~$X$, $S\coloneqq\sqrt{Z}$ and $\sqrt{\mathrm{i}}\Identity$.
Similar to the \textsc{PSF}, we consider states that are invariant under the action of such generalized stabilizer operators.
The resulting stabilizer formalism is called here the \emph{XS-stabilizer formalism}.
It is a subclass of the monomial stabilizer formalism introduced recently in \cite{Vandennest2011monomial}.
Interestingly, the XS-stabilizer formalism allows for non-Abelian stabilizer groups, whereas it is well known that stabilizer groups in the \textsc{PSF} must be Abelian.
 
Even though the definition of the XS-stabilizer formalism is close to that of the original \textsc{PSF},  these frameworks differ in several ways.
In particular, the XS-stabilizer formalism is considerably richer than the \textsc{PSF}, and we will encounter several manifestations of this.
At the same time, the XS-stabilizer formalism keeps many favorable features of the \textsc{PSF}.
For example, XS-stabilizer groups have a simple structure and are easy to manipulate, and there exists a close relation between the stabilizer generators of an XS-stabilizer state/code and the associated Hamiltonian.
Moreover, we will show that (under a mild restriction of the XS-stabilizers) many quantities of interest can be computed efficiently, such as expectation values of local observables, code degeneracy and logical operators.
However, in most cases we found that efficient algorithms could not be obtained by straightforwardly extending methods from the \textsc{PSF}, and new techniques needed to be developed.
 
The purpose of this paper is to introduce the XS-stabilizer formalism, to provide examples of XS-stabilizer states and codes that are not covered by the \textsc{PSF}  and to initiate a systematic development of the XS-stabilizer framework.
In particular, we discuss several properties related to the structure of XS-stabilizer states and codes, their entanglement, their efficient generation by means of quantum circuits and their efficient simulation with classical algorithms.
A detailed statement of our results is given in section \ref{sec:statement of main results}.
Here we briefly highlight two aspects.
 
First, we consider the potential of the XS-stabilizer formalism to describe topological phases.
This is motivated by recent works on classifying quantum phases within the \textsc{PSF}~\cite{Yoshida2011classification,Bombin2012universal}, which is related to the problem of finding a self-correcting quantum memory.
In particular, Haah constructed a novel Pauli stabilizer code for a 3D lattice in~\cite{Haah2011local} and gave evidence that it might be a self-correcting quantum memory even at non-zero temperature.
In the present paper we show that the XS-stabilizer formalism can describe 2D topological phases beyond the \textsc{PSF} and, surprisingly, some of these harbour non-Abelian anyons.
Specific examples of models covered by the XS-stabilizer formalism are the doubled semion model~\cite{LevinWen} and, more generally, the twisted quantum double models for the groups~$\mathbb{Z}_2^k$~\cite{Kitaev2006,Hu2013twisted,Buerschaper}.
 
Second, we study entanglement in the XS-stabilizer formalism.
Various entanglement properties of Pauli stabilizer states have been studied extensively in the past decade~\cite{Hein2004multiparty, Fattal2004entanglement}.
While the bipartite entanglement structure is very well understood, less is known about the multipartite scenario.
For example, recently in Ref.~\cite{Linden2013quantum} the entropy inequalities for Pauli stabilizer states were studied.
Here we will show that, for any bipartition, we can always map any XS-stabilizer state into a Pauli stabilizer state locally, which means their bipartite entanglement is identical.
This implies in particular that all reduced density operators of an XS-stabilizer state are projectors and each single qubit is either fully entangled with the rest of the system or fully disentangled from it.
In contrast, the XS-stabilizer formalism is genuinely richer than the \textsc{PSF} when viewed through the lens of multipartite entanglement.
For example, we will show that there exist XS-stabilizer states that cannot be mapped onto any Pauli stabilizer state under local unitary operations.
Thus there seems to be a complex and intriguing relation between the entanglement properties of Pauli and XS-stabilizer states.
 
We also mention other works that, similar in spirit to the present paper, aim at extending the \textsc{PSF}.
These include: Ref.~\cite{Hartmann2007weighted} which introduced the family of weighted graph states as generalizations of graph and stabilizer states; Ref.~\cite{Kruszynska2009local} where the family of locally maximally entanglable (\textsc{LME}) states were considered (which in turn generalize weighted graph states);  Ref.~\cite{Rossi2013quantum} where hypergraph states were considered.
The XS-stabilizer formalism differs from the aforementioned state families in that its starting point is the representation of states by their stabilizer operators.
We have not yet investigated the potential interrelations between these classes, but it would be interesting to understand this in more detail.

{\bf Outline of the paper.}
Readers who are mainly interested in an overview of our results, rather than in the technical details, may want to focus on sections \ref{sec:the xs-stabilizer formalism} and \ref{sec:statement of main results}.
In section \ref{sec:the xs-stabilizer formalism} we introduce the basic notions of XS-stabilizer states and codes.
In section \ref{sec:statement of main results} we give a summary of the results presented in this paper.
The following sections are dedicated to developing the technical arguments.

\section{The XS-Stabilizer Formalism}
\label{sec:the xs-stabilizer formalism}

In this section we introduce the basic notions of XS-stabilizer states and codes and we provide several examples.

\subsection{Definition}

First we briefly recall the standard Pauli stabilizer formalism.
Let~$X$, $Y$ and~$Z$ be the standard Pauli matrices.
The single-qubit Pauli group is $\gen{\mathrm{i}\Identity,X,Z}$.
For a system consisting of $n$~qubits we use~$X_j$, $Y_j$ and~$Z_j$ to represent the Pauli matrices on the $j$-th qubit.
An operator~$g$ on $n$~qubits is a Pauli operator if it has the form $g=g^{(1)}\otimes\dots\otimes g^{(n)}$ where each~$g^{(i)}$ belongs to the single-qubit Pauli group.
Every $n$-qubit Pauli operator can be written as
\begin{equation}
    g
    =\mathrm{i}^s
     X^{a_1}
     Z^{b_1}\otimes
     \dots\otimes
     X^{a_n}
     Z^{b_n}
\end{equation}
where $s\in\set{0,\dots,3}$, $a_j\in\set{0,1}$ and $b_j\in\set{0,1}$.
We say an $n$-qubit quantum state~$\ket{\psi}\neq 0$ is stabilized by a set of Pauli operators~$\set{g_j}$ if
\begin{equation}
    g_j\mkern2mu
    \ket{\psi}
    =\ket{\psi}
    \qquad
    \text{for all~$j$}.
\end{equation}
The operators~$g_j$ are called \emph{stabilizer operators} of~$\ket{\psi}$.

In this paper, we generalize the Pauli stabilizer formalism by allowing more general stabilizer operators.
Instead of the single-qubit Pauli group, we start from the larger group $\PauliSGroup\coloneqq\gen{\alpha\Identity,X,S}$ where $\alpha=\mathrm{e}^{\mathrm{i}\pi/4}$ and $S=\diag(1,\mathrm{i})$.
Note that the latter group, which we call the \emph{Pauli-S group}, contains the single-qubit Pauli group since $S^2=Z$.
We then consider stabilizer operators $g=g^{(1)}\otimes\dots\otimes g^{(n)}$ where each~$g^{(i)}$ is an element of~$\PauliSGroup$.
It is easy to show
that every such operator can be written as
\begin{equation}
    \label{eq:expansion}
    g
    =\alpha^s
     X^{a_1}
     S^{b_1}\otimes
     \dots\otimes
     X^{a_n}
     S^{b_n}
    \eqqcolon
     \alpha^s
     X(\vec{a}\mkern1mu)\mkern2mu
     S(\vec{b}\mkern1mu)
\end{equation}
where $s\in\set{0,\dots,7}$, $a_j\in\set{0,1}$ and $b_j\in\set{0,\dots,3}$.
Here we also defined $X(\vec{a}\mkern1mu)\coloneqq X^{a_1}\otimes\dots\otimes X^{a_n}$ for $\vec{a}=(a_1,\dots,a_n)$ and similarly $S(\vec{b}\mkern1mu)$ and~$Z(\vec{c}\mkern3mu)$.
These are called \emph{X-type}, \emph{S-type} and \emph{Z-type operators} respectively.

For a set~$\set{g_1,\dots,g_m}$ of such operators we consider the group $G=\gen{g_1,\dots,g_m}$, and we say a state~$\ket{\psi}\neq 0$ is stabilized by~$G$ if we have $g\mkern2mu\ket{\psi}=\ket{\psi}$ for every~$g\in G$.
Whenever such a state exists we call~$G$ an \emph{XS-stabilizer group}.
The space~$\XSCode$ of all states stabilized by~$G$ is referred to as the \emph{XS-stabilizer code} associated with~$G$.
A state which is uniquely stabilized by~$G$ is called an \emph{XS-stabilizer state}.

Thus the XS-stabilizer formalism is a generalization of the Pauli stabilizer formalism.
Perhaps the most striking difference is that XS-stabilizer states/codes may have a \emph{non}-Abelian XS-stabilizer group~$G$ -- while Pauli stabilizer groups must always be Abelian.
We will see examples of this in the next section.

\subsection{Examples}

Here we give several examples of XS-stabilizer states and codes and highlight how their properties differ from the standard Pauli stabilizer formalism.

A first simple example of an XS-stabilizer state is the 6-qubit state $|\psi\rangle$ stabilized by the (non-commuting) operators
\begin{equation}
    \begin{split}
        g_1
        & =X\otimes
           S^3\otimes
           S^3\otimes
           S\otimes
           X\otimes
           X, \\
        g_2
        & =S^3\otimes
           X\otimes
           S^3\otimes
           X\otimes
           S\otimes
           X, \\
        g_3
        & =S^3\otimes
           S^3\otimes
           X\otimes
           X\otimes
           X\otimes
           S.
    \end{split}
\end{equation}
Explicitly, $|\psi\rangle$ is given by
\begin{equation}
    \label{eq:6qubit_example}
    \ket{\psi}
    =\smashoperator[l]{\sum_{x_j
                             =0}^1}
     (-1)^{x_1
           x_2
           x_3}\mkern2mu
     \ket{x_1,
          x_2,
          x_3,
          x_1\oplus
          x_2,
          x_2\oplus
          x_3,
          x_3\oplus
          x_1}.
\end{equation}
It is straightforward to show that $|\psi\rangle$ is the unique (up to a global phase) state stabilized by $g_1$, $g_2$ and $g_3$.
Note that in this example 3 stabilizer operators suffice to uniquely determine the 6-qubit state $|\psi\rangle$.
This is different from the Pauli stabilizer formalism, where 6 stabilizers would be necessary (being equal to the number of qubits).
Notice also that $|\psi\rangle$ contains amplitudes of the form $(-1)^{c(x)}$ where $c(x)$ is a cubic polynomial of the bit string $x=(x_1, x_2, x_3)$.
This shows that $|\psi\rangle$ cannot be a Pauli stabilizer state, since the latter cannot have such cubic amplitudes~\cite{Dehaene2003the}.
This example thus shows that the XS-stabilizer formalism covers a strictly larger set of states than the Pauli stabilizer formalism.
What is more, we will show (cf. section \ref{sec:lu non equivalence}) that the state $|\psi\rangle$ is not equivalent to any Pauli stabilizer state even if arbitrary local basis changes are allowed.
Thus, $|\psi\rangle$ belongs to a different local unitary equivalence class than any Pauli stabilizer state.

\begin{figure}
    \centering
    \includegraphics{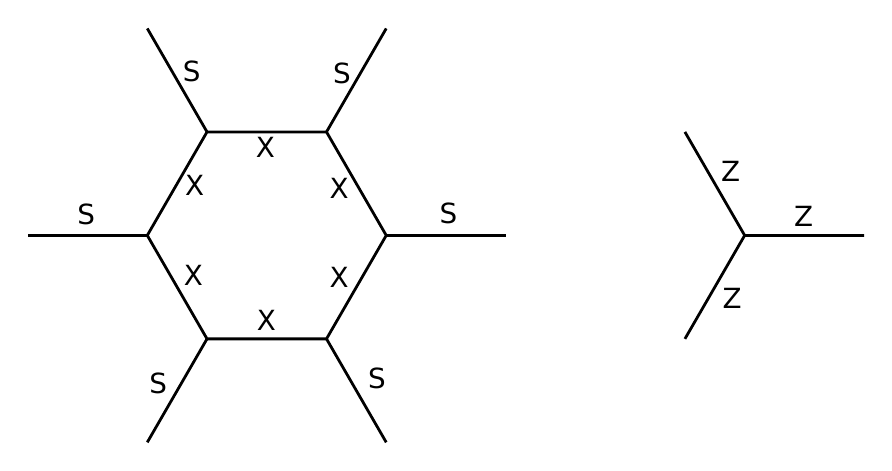}
    \caption{In the doubled semion model, the qubits are on the edges of a honeycomb lattice.
        The ground space of the Hamiltonian can be equivalently described by the two types of XS-stabilizers in the above figure.
        The left one is associated to each face of the lattice and the right one is associated to each vertex.}
    \label{fig:doubled semion}
\end{figure}

A second example is the \emph{doubled semion model} which belongs to the family of string-net models~\cite{LevinWen}.
It is defined on a honeycomb lattice with one qubit per edge and has two types of stabilizer operators%
\footnote{The local single-qubit basis used in~\cite{LevinWen} is different from ours.}
which are shown in Figure~\ref{fig:doubled semion}.
Let~$g_s$ and~$g_p$ be the stabilizer operators corresponding to the vertex~$s$ and the face~$p$ respectively.
Then the ground space of the doubled semion model consists of all states~$\ket{\psi}$ satisfying $g_s\mkern2mu\ket{\psi}=g_p\mkern2mu\ket{\psi}=\ket{\psi}$ for all~$s$ and~$p$.
The doubled semion model is closely related to the toric code which is a Pauli stabilizer code.
The Pauli stabilizer operators of the toric code are obtained from the XS-stabilizer operators of the doubled semion model by replacing all occurrences of~$S$ with~$\Identity$.
This is no coincidence since both the doubled semion model and the toric code are \emph{twisted quantum double models} for the group~$\mathbb{Z}_2$~\cite{Kitaev2006,Hu2013twisted,Buerschaper}.
In spite of this similarity it is known that both models represent different topological phases~\cite{LevinWen}.
Thus, the XS-stabilizer formalism allows one to describe states with genuinely different topological properties compared to any state arising in the Pauli stabilizer formalism~\cite{Yoshida2011classification,Bombin2012universal}.
In fact, we can use XS-stabilizers to describe other, more complex, twisted quantum double models as well, as we will show in section~\ref{sec:twisted_quantum_double}.
Some of these even support non-Abelian anyons.

The third example is related to magic state distillation.
In \cite{Bravyi2005universal} the authors consider a 15 qubit code $\CSSCode{Z}{\ReedMullerCode_2}{XS}{\ReedMullerCode_1}$, where $\ReedMullerCode_1$ and~$\ReedMullerCode_2$ are punctured Reed-Muller codes of order one and two, respectively.
Roughly speaking, this quantum code is built from two types of generators.
One type has the form $Z\otimes\dots\otimes Z$ acting on some of the qubits, while the other type has the form $XS\otimes\dots\otimes XS$.
Surprisingly, this 15 qubit XS-stabilizer code has the same code subspace as the Pauli stabilizer code $\CSSCode{Z}{\ReedMullerCode_2}{X}{\ReedMullerCode_1}$ which is obtained by replacing every $S$ operator with an identity matrix.
From this example we can see that having $S$ in the stabilizer operators does not necessarily mean an XS-stabilizer group and a Pauli stabilizer group stabilize different spaces.

\begin{table}
    \begin{center}
        \begin{tabular}{@{}lccc@{}}
            \toprule
                                             &        Pauli        &      Regular XS     &           General XS          \\
            \midrule
            Commuting stabilizer operators   &         yes         &          no         &               no              \\
            Commuting parent Hamiltonian     &         yes         &         yes         &              yes              \\
            Complexity of stabilizer problem & \ComplexityClass{P} & \ComplexityClass{P} & \ComplexityClass{NP}-complete \\
            Non-Abelian anyons in 2D         &          no         &         yes         &              yes              \\
            \bottomrule
        \end{tabular}
    \end{center}
    \caption{Summary of the properties}
\end{table}

\section{Main Results}
\label{sec:statement of main results}

\subsection{Commuting Parent Hamiltonian}

Even though an XS-stabilizer group $G=\gen{g_1,\dots,g_m}$ is non-Abelian in general, we will show that there always exists a Hamiltonian $H'=\sum_j h_j$ with mutually commuting projectors~$h_j$ whose ground state space coincides with the space stabilized by~$G$ (section \ref{sec:hamiltonians}).
If the generators of~$G$ satisfy some locality condition (e.g. they are $k$-local on some lattice), then the~$h_j$ will satisfy the same locality condition (up to a constant factor).
This means that general properties of ground states of commuting Hamiltonians apply to XS-stabilizer states.
For example, every state uniquely
stabilized by a set of local XS-stabilizers defined on a $D$-dimensional lattice satisfies the area law \cite{Wolf2008area}, and for local XS-stabilizers on a 2D lattice, we can find string like logical operators \cite{Bravyi2010tradeoffs}.

While the ground state spaces of~$H'$ and the non-commuting Hamiltonian $H=\sum_{j=1}^m(g_j+g_j^\dagger)$ are identical, the latter may have a completely different spectrum.
This may turn out important for the purpose of quantum error correction.

\subsection{Computational Complexity of Finding Stabilized States}

In the Pauli stabilizer formalism, it is always computationally easy to determine whether, for a given set of stabilizer operators, there exists a common stabilized state.
However, we will prove that the same question is \ComplexityClass{NP}-\emph{complete} for XS-stabilizers (see section \ref{sec:complexity}).
More precisely, we consider the problem \Problem{XS-Stabilizer} defined as follows: given a set of XS-stabilizer operators $\{g_1,\dots,g_m\}$, the task is to decide whether there exists a state $\ket{\psi}\neq 0$ stabilized by every~$g_j$.
The \ComplexityClass{NP}-hardness part of the \Problem{XS-Stabilizer} problem is proved via a reduction from the \Problem{Positive 1-in-3-Sat} problem.
In order to show that the problem is in~$\ComplexityClass{NP}$, we use tools developed for analyzing monomial stabilizers, as introduced in~\cite{Vandennest2011monomial}.

The \ComplexityClass{NP}-hardness of the \Problem{XS-Stabilizer} problem partially stems from the fact that the group $G= \langle g_1,\dots,g_m\rangle$ may contain \emph{diagonal} operators which have one or more $S$~operators in their tensor product representation~\eqref{eq:expansion}.
In order to render the \Problem{XS-Stabilizer} problem tractable, we impose a (mild) restriction on the group~$G$ and demand that every diagonal operator in~$G$ can be written as a tensor product of~$\Identity$ and~$Z$, i.e. no diagonal operator in~$G$ may contain an $S$~operator.
We call such a group $G$ \emph{regular}.
We will show that, for every regular $G$, the existence of a state stabilized by $G$ can then be checked efficiently (section \ref{sec:regular}).

Finally, we will show that in fact every XS-stabilizer \emph{state} affords a regular stabilizer group (although finding it may be computationally hard), i.e. the condition of regularity does not restrict the set of states that can be described by the XS-stabilizer formalism (Section~\ref{sec:non regular xs stabilizer}).
In contrast, the stabilizer group of an XS-stabilizer \emph{code} cannot always be chosen to be regular.

\subsection{Entanglement}

Given an XS-stabilizer state $\ket{\psi}$ with associated XS-stabilizer group $G$, we show how to compute the entanglement entropy for any bipartition $(A, B)$ (section \ref{sec:entanglement}).
This is achieved by showing that $\ket{\psi}$ can always be transformed into a Pauli stabilizer state $|\phi_{A ,B}\rangle$ (which depends on the bipartition in question) by applying a unitary $U_A\otimes U_B$, where $U_A$ and $U_B$ each only act on the qubits in each party.
Since an algorithm to compute the entanglement entropy of Pauli stabilizer states is known, this yields an algorithm to compute this quantity for the original XS-stabilizer state~$\ket{\psi}$ since the unitary $U_A\otimes U_B$ does not change the entanglement.
Our overall algorithm is efficient (i.e. runs in polynomial time in the number of qubits) for all regular XS-stabilizer groups (cf. also section \ref{sec:algorithms}).
It is worth noting that our method of computing the entanglement entropy uses a very different technique compared to the one typically used for studying the entanglement entropy of Pauli stabilizer states (for example, the methods in~\cite{Linden2013quantum}).

The fact that $\ket{\psi}=U_A\otimes U_B\mkern2mu\ket{\phi_{A,B}}$ for any bipartition~$(A,B)$ implies in particular that any reduced density matrix of~$\ket{\psi}$ is a projector since this is the case for all Pauli stabilizer states~\cite{Hein2006entanglement}.
Consequently, all $\alpha$-Rényi entanglement entropies of an XS-stabilizer state coincide with the logarithm of the Schmidt rank.

We also formulate the following open problem: for every XS-stabilizer state~$\ket{\psi}$, does there exist a single Pauli stabilizer state~$\ket{\phi}$ with the same Schmidt rank as~$\ket{\psi}$ for \emph{every} bipartition? For example, it would be interesting to know whether the inequalities in \cite{Linden2013quantum} hold for XS-stabilizer states.

As far as multipartite entanglement is concerned, we finally show that the 6-qubit XS-stabilizer state~\eqref{eq:6qubit_example} is \emph{not} equivalent to any Pauli stabilizer state even if arbitrary local basis changes are allowed.

\subsection{Efficient Algorithms}

In section~\ref{sec:algorithms} we show that several basic tasks can be solved efficiently for an XS-stabilizer state~$\ket{\psi}$, provided its regular XS-stabilizer group is known:
\begin{enumerate}
    \item Compute the entanglement entropy for any bipartition.
    \item Compute the expectation value of any local observable.
    \item Prepare $\ket{\psi}$ on a quantum computer with a poly-size quantum circuit.
    \item Compute the function $f(x)$ in the standard basis expansion
        \begin{equation}
            \ket{\psi}
            =\sum_x
             f(x)\mkern2mu
             \ket{x}.
        \end{equation}
\end{enumerate}
Moreover, we can efficiently construct a basis $\set{\ket{\psi_1},\dots,\ket{\psi_d}}$ for any XS-stabilizer code with a \emph{regular} XS-stabilizer group.
In particular, we can efficiently compute the degeneracy~$d$ of the code.
For each~$\ket{\psi_j}$ we can again solve all the above tasks efficiently.
Finally, we can also efficiently compute logical operators.

The algorithms given in section \ref{sec:algorithms} depend heavily on the technical results for XS-stabilizer states and codes given in section \ref{sec:constructing_basis}, where we characterize several structural properties of these states and codes.

\section{Basic Group Theory}
\label{sec:diagonal subgroup}

In this section we introduce some further basic notions, discuss basic manipulations of XS-stabilizer operators and describe some important subsets and subgroups of XS-stabilizer groups.

\subsection{Pauli-S Group}

Let us write
\begin{equation}
    [g,
     h]
    \coloneqq
     g
     h
     g^{-1}
     h^{-1}
\end{equation}
for the commutator of any two group elements~$g$ and~$h$.
In the following we always assume the elements of a set $\set{g_1,\dots,g_m}\subset\PauliSGroupMany$ to be given in the standard form
\begin{equation}
    \label{eq:standard_form}
    g_j
    =\alpha^{s_j}
     X(\vec{a}_j)\mkern2mu
     S(\vec{b}_j).
\end{equation}

\begin{lemma}[Commutators]
    \label{lem:commutators}
    \begin{equation}
        [g_1,
         g_2]
        =\bigotimes_{k
                     =1}^n\mkern3mu
         (-1)^{a_{1
                  k}
               a_{2
                  k}
               (b_{1
                   k}+
                b_{2
                   k})}\mkern2mu
         (\mathrm{i}
          Z_k)^{a_{1
                   k}
                b_{2
                   k}-
                a_{2
                   k}
                b_{1
                   k}}.
    \end{equation}
\end{lemma}
\begin{proof}
    It suffices to prove this for~$\PauliSGroup$ and $\alpha^s=1$.
    So let $g_j=X^{a_j}S^{b_j}$.
    Then
    \begin{equation}
        g_2
        g_1
        =(-1)^{a_1
               a_2
               (b_1+
                b_2)}\mkern2mu
         (\mathrm{i}
          Z)^{a_2
              b_1-
              a_1
              b_2}
         g_1
         g_2
    \end{equation}
    where we used $S^bX^a=(-\mathrm{i}Z)^{ab}X^aS^b$.
    The claim for the tensor product group~$\PauliSGroupMany$ follows from applying the above to each component.
\end{proof}

\begin{lemma}[Squares]
    \label{lem:squares}
    Let $g=\alpha^s X^a S^b\in \PauliSGroup$.
    Then
    \begin{equation}
        g^2
        =\mathrm{i}^{s+
                     a
                     b}\mkern2mu
         Z^{(a+
             1)
            b}.
    \end{equation}
\end{lemma}

\begin{lemma}[Multiplication]
    \label{lem:multiplication}
    There exists~$\vec{b}'$ such that
    \begin{equation}
        g_1
        g_2
        \propto
         X(\vec{a}_1\oplus
           \vec{a}_2)\mkern2mu
         S(\vec{b}').
    \end{equation}
\end{lemma}

\subsection{Important Subgroups}

For any group~$G\subset\PauliSGroupMany$ there are two important subgroups.

\begin{definition}
    The group
    \begin{equation}
        \DiagonalSubgroup
        \coloneqq
         G\cap
         \gen{\alpha
              \Identity,
              S_1,
              \dots,
              S_n}
        =G\cap
         \setbuilder[\big]{\alpha^s
                           S(\vec{b}\mkern1mu)}
                          {s\in
                           \set{0,
                                \dots,
                                7},
                           \vec{b}\in
                           \set{0,
                                \dots
                                3}^n}
    \end{equation}
    is called the \emph{diagonal subgroup} and
    \begin{equation}
        \label{eq:Z-subgroup}
        \ZSubgroup
        \coloneqq
         G\cap
         \gen{\alpha
              \Identity,
              Z_1,
              \dots,
              Z_n}
        =G\cap
         \setbuilder[\big]{\alpha^s
                           Z(\vec{c}\mkern3mu)}
                          {s\in
                           \set{0,
                                \dots,
                                7},
                           \vec{c}\in
                           \set{0,
                                1}^n}
    \end{equation}
    is called the \emph{Z-subgroup}.
\end{definition}

In other words, the diagonal subgroup~$\DiagonalSubgroup$ contains all elements of~$G$ which are diagonal matrices in the computational basis.
These are precisely the elements which do not contain any $X$~operators in their tensor product representation~\eqref{eq:expansion}.
The Z-subgroup~$\ZSubgroup$ consists of all Z-type operators.
In particular, all commutators and squares of elements in~$G$ are contained in~$\ZSubgroup$, as can be seen from Lemmas~\ref{lem:commutators} and~\ref{lem:squares}.

If $G$ is an XS-stabilizer group, then all its elements must have an eigenvalue~$1$.
Clearly, its Z-subgroup~$\ZSubgroup$ must then be contained in~$\gen{\pm Z_1,\dots,\pm Z_n}\setminus\set{-\Identity}$, otherwise $\ZSubgroup$ (and thus~$G$) may contain elements which lack the eigenvalue~$1$, as is evident from~\eqref{eq:Z-subgroup}.
In particular, $G$ cannot contain~$-\Identity$.
This implies that $\ZSubgroup$ lies in the centre~$Z(G)$ of~$G$.
Indeed, every $\MagicZ\in\ZSubgroup$ either commutes or anticommutes with all elements of~$G$,
however, $[\MagicZ,g]=-\Identity\in G$ for some $g\in G$ would give a contradiction.
Furthermore one can easily see from the above that all elements of~$\ZSubgroup$ have an order of at most~$2$, thus we conclude that $g^4=\Identity$ for all $g\in G$ since $g^2\in\ZSubgroup$.
We have just proved

\begin{proposition}
    \label{prop:stabilizer_group}
    Every XS-stabilizer group~$G$ satisfies
    \begin{enumerate}
        \item $-\Identity\not\in G$,
        \item $\ZSubgroup\subset\gen{\pm Z_1,\dots, \pm Z_n}\setminus\set{-\Identity}=\set{(-1)^s Z(\vec{c}\mkern3mu)}\setminus\set{-\Identity}$,
        \item $\ZSubgroup\subset Z(G)$,
        \item $g^4=\Identity$ for all $g\in G$.
    \end{enumerate}
\end{proposition}

\subsection{Admissible Generating Sets}

Typically it is computationally hard to check the above necessary conditions for the \emph{entire} group~$G$.
Instead, we focus on a small set of generators which fully determine~$G$, like in the Pauli stabilizer formalism.
We are interested in finding necessary conditions for such a set to generate an XS-stabilizer group.

While we can build arbitrary words from the generators, of course, commutators and squares of generators will play a distinguished role in this article.

\begin{definition}
    Let $\Generators=\set{g_1,\dots,g_m}\subset\PauliSGroupMany$.
    Then
    \begin{align}
        \Commutators
        & \coloneqq
           \setbuilder{[g_j,
                        g_k]}
                      {g_j,
                       g_k\in
                       \Generators\wedge
                       j
                       \neq
                        k}, \\
        \Squares
        & \coloneqq
           \setbuilder{g_j^2}
                      {g_j\in
                       \Generators}.
    \end{align}
\end{definition}

\begin{definition}
    A set $\Generators=\set{g_1,\dots,g_m}\subset\PauliSGroupMany$ is called an \emph{admissible generating set} if
    \begin{enumerate}
        \item every~$g_j$ has an eigenvalue~$1$,
        \item every~$[g_j,g_k]$ has an eigenvalue~$1$,
        \item $[[g_j,g_k],g_l]=\Identity$,
        \item $[g_j^2,g_k]=\Identity$.
    \end{enumerate}
\end{definition}

Clearly, if $G=\gen{\Generators}$ is an XS-stabilizer group, then $\Generators$ must be an admissible generating set by Proposition~\ref{prop:stabilizer_group} (and the discussion preceding it).
The converse is not true: there exist admissible generating sets~$\Generators$ for which $\langle \Generators\rangle$ is \emph{not} an XS-stabilizer group.

Note that the properties in the above definition are independent in the sense that the first $k$~properties do not imply the next one.
It can be checked in $\poly(n,m)$ time whether a given generating set~$\Generators$ is admissible.

We then have the following lemma:

\begin{lemma}[Relative standard form]
    \label{lem:relative_standard_form}
    If $\Generators=\set{g_1,\dots,g_m}\subset\PauliSGroupMany$ is an admissible generating set, then the elements of~$G=\gen{\Generators}$ are given by
    \begin{equation}
        \label{eq:relative_standard_form}
        \MagicZ\mkern2mu
        g(\vec{x}\mkern1mu)
        \coloneqq
         \MagicZ
         g_1^{x_1}
         \cdots
         g_m^{x_m}
    \end{equation}
    where $\vec{x}\in\mathbb{Z}_2^m$ and $\MagicZ\in\gen{\Commutators\cup\Squares}\subset\ZSubgroup$.
    
    Furthermore, for two elements $h=\MagicZ\mkern2mu g(\vec{x}\mkern1mu)$ and $h'=\MagicZ' g(\vec{x}\mkern1mu')$ we have
    \begin{equation}
        h
        h'
        =\MagicZ''
         g(\vec{x}\oplus
           \vec{x}\mkern1mu').
    \end{equation}
\end{lemma}
\begin{proof}
    Let $h=g_{\beta_1}g_{\beta_2}\cdots g_{\beta_p}\in G$ an arbitrary word in the generators~$\Generators$.
    We will show how to reduce it to the form~\eqref{eq:relative_standard_form}.
    Suppose $\beta_{j-1}>\beta_j$ for some~$j$.
    Since $g_{\beta_{j-1}}g_{\beta_j}=\MagicZ g_{\beta_j}g_{\beta_{j-1}}$ for some~$\MagicZ\in\Commutators$ we can reorder the generators locally and move any commutator~$\MagicZ$ to the left.
    (Since $\Generators$ is admissible, $\MagicZ$ commutes with all generators.)
    Repeating this procedure we arrive at $h=\MagicZ g_1^{x_1}\cdots g_m^{x_m}$ for some~$\MagicZ\in\gen{\Commutators}$, where the exponents~$x_j$ may still be arbitrary integers.
    We can restrict them to~$\{0,1\}$ by extracting squares of generators and moving them to the left.
    We obtain $h=\MagicZ\MagicZ'g_1^{x_1}\cdots g_m^{x_m}$ for some $\MagicZ'\in\gen{\Squares}$ which proves the first claim.
    The second claim follows easily from a similar argument.
\end{proof}

The diagonal subgroup~$\DiagonalSubgroup$ will play an important role in this paper.
Here we give a method to compute the generators of the diagonal subgroup~$\DiagonalSubgroup$ efficiently.

\begin{lemma}
    \label{lem:generators of diagonal subgroup}
    If $\Generators=\set{g_1,\dots,g_m}\subset\PauliSGroupMany$ is an admissible generating set and $G=\gen{\Generators}$, then a generating set of~$\DiagonalSubgroup$ can be found in $\poly(n,m)$ time.
\end{lemma}
\begin{proof}
    We see from Lemma~\ref{lem:relative_standard_form} that $\DiagonalSubgroup$ is generated by $\Commutators$, $\Squares$ and those elements~$g(\vec{x}\mkern1mu)$ which are diagonal.
    Hence we only need to find a generating set for the latter.
    Assume that the generators of~$G$ are given in the standard form~\eqref{eq:standard_form} and define the $n\times m$ matrix
    \begin{equation}
        A
        \coloneqq
         [\begin{matrix}
              \vec{a}_1 & \dots & \vec{a}_m
          \end{matrix}]
    \end{equation}
    whose columns are the bit strings~$\vec{a}_j$.
    It follows from Lemma~\ref{lem:multiplication} that $g(\vec{x}\mkern1mu)\propto X(A\vec{x}\mkern1mu)\mkern2mu S(\vec{b}')$ for some~$\vec{b}'$.
    This implies that $g(\vec{x}\mkern1mu)$ is a diagonal operator if and only if $A\vec{x}=0$ over~$\mathbb{Z}_2$.
    Denote a basis of the solution space of this linear system by $\set{\vec{u}_i}$.
    Such a basis can be computed efficiently.
    Notice that by Lemma~\ref{lem:relative_standard_form} we have $g(\vec{u}_i\oplus\vec{u}_j)=\MagicZ\mkern2mu g(\vec{u}_i)\mkern2mu g(\vec{u}_j)$ for any two basis vectors~$\vec{u}_i$ and~$\vec{u}_j$ and some $\MagicZ\in\gen{\Commutators\cup\Squares}$.
    This implies that all diagonal elements~$g(\vec{x}\mkern1mu)$ can be generated by $\Commutators$, $\Squares$ and~$\set{g(\vec{u}_i)}$, and so can~$\DiagonalSubgroup$.
    Finally we note that the length of this generating set is $\poly(m,n)$.
\end{proof}

\section{Commuting Parent Hamiltonian}
\label{sec:hamiltonians}

In this section we show that the space stabilized by $\{g_j\}$ can also be described by the ground space of a set of commuting Hamiltonians.
In fact, the Hamiltonians are monomial.

Let $G=\gen{\Generators}$ be an XS-stabilizer group with the generators~$\Generators=\set{g_1,\dots,g_m}$ and the corresponding code~$\XSCode$.
While it is straightforward to turn each generator into a Hermitian projector
onto its stabilized subspace, these projectors will \emph{not} commute with each other in general.
Perhaps surprisingly, we can still construct a commuting parent Hamiltonian for~$\XSCode$ by judiciously choosing a subset of~$G$ such that a) this subset yields a commuting Hamiltonian with the \emph{larger} ground state space~$\mathcal{L}\supset\XSCode$, and b) all generators mutually commute when restricted to~$\mathcal{L}$.
We will call~$\mathcal{L}$ the \emph{gauge-invariant subspace} in the following.

We claim that the subset~$\Commutators\cup\Squares\subset G$ precisely fits this strategy.
First, let us define $P_g\coloneqq(\Identity+g)/2$ for arbitrary $g\in G$.
It is easy to see that all~$P_\MagicZ$ with $\MagicZ\in\Commutators\cup\Squares$ are Hermitian projectors which commute with each other and all elements of~$G$.
We may define the gauge-invariant subspace as the image of the Hermitian projector $P\coloneqq\prod_\MagicZ P_\MagicZ$ which commutes with all~$P_\MagicZ$ and all elements of~$G$ by construction.
Moreover, note that
\begin{equation}
    \label{eq:sponge}
    P
    \MagicZ
    =P,
\end{equation}
in other words, the gauge-invariant subspace \enquote{absorbs} commutators and squares of generators.
Second, it is easy to check that all~$P P_{g_j}$ with $g_j\in\Generators$ are Hermitian projectors which mutually commute.
Indeed, they are projectors since $(P P_{g_j})^2=(P^2+2P^2 g_j+P^2 g_j^2)/4=P P_{g_j}$
where we used~\eqref{eq:sponge}.
Moreover, they are Hermitian since $(P g_j)^\dagger=g_j^3 P=P g_j$
where we used Proposition~\ref{prop:stabilizer_group} and~\eqref{eq:sponge}.
Finally, they commute with each other because 
\begin{equation}
    P
    g_k
    P
    g_j
    =P
     g_k
     g_j
    =P
     \MagicZ
     g_j
     g_k
    =P
     g_j
     g_k
    =P
     g_j
     P
     g_k
\end{equation}
for some $\MagicZ\in\Commutators$ which is absorbed by virtue of~\eqref{eq:sponge}.

We can now define the commuting Hamiltonian associated with~$G$ (and~$\Generators$) by
\begin{equation}
    \label{eq:commuting_Hamiltonian}
    H_{G,
       \Generators}
    \coloneqq
     \sum_\MagicZ
     (\Identity-
      P_\MagicZ)+
     \smashoperator[l]{\sum_{g_j\in
                             \Generators}}
     (\Identity-
      P
      P_{g_j}).
\end{equation}
It remains to show that the space annihilated by~$H_{G,\Generators}$ is precisely the XS-stabilizer code~$\XSCode$.
It is easy to see that a state~$\ket{\psi}$ has zero energy if it is stabilized by~$G$.
Conversely, if $\ket{\psi}$ has zero energy then $P_\MagicZ\mkern2mu\ket{\psi}=\ket{\psi}$ and $PP_{g_j}\mkern2mu\ket{\psi}=\ket{\psi}$ follow directly.
The former condition actually implies $P\mkern2mu\ket{\psi}=\ket{\psi}$, hence the latter turns into $P_{g_j}\ket{\psi}=\ket{\psi}$ from which we deduce $g_j\mkern2mu\ket{\psi}=\ket{\psi}$.

\begin{remark}[Locality]
    It is not hard to see the above construction of a commuting Hamiltonian can be modified to preserve the locality of $g_j$.
    Assume $g_j$ is local on a $d$-dimension lattice.
    Then by construction, $P_{\MagicZ}$ are also local.
    Thus the only nonlocal terms in the Hamiltonian are $PP_{g_j}$, and below we show how to make a modification such that they become local.
    We say $g_k$ is a neighbour of $g_j$ if $g_j$ and $g_k$ act on some common qubits, and we denote that by $k\in n(j)$ (we also set $j\in n(j)$ for our purpose).
    It is easy to check that if we replace the $PP_{g_j}$ terms in the Hamiltonian by
    \begin{equation}
        \Bigl(\prod_{k\in
                     n(j)}
              P_{j
                 k}\Bigr)
        P_{g_j},
    \end{equation}
    the Hamiltonian is still commuting, while it is now local on the lattice.
\end{remark}

\begin{remark}[Quantum error correcting code]
    We can use XS-stabilizer codes~$\XSCode$ for quantum error correction.
    Here it is important that error syndromes can be measured \emph{simultaneously} which seems impossible if the XS-stabilizer group~$G$ is non-Abelian.
    Yet we can exploit the commuting stabilizers constructed above and extract the error syndromes in two rounds.
    First we measure the syndromes of the mutually commuting stabilizers in the subset~$\Commutators\cup\Squares$ and correct as necessary.
    We are now guaranteed to be in the gauge-invariant subspace where the original generators~$\set{g_j}$ commute.
    We can thus measure their syndromes simultaneously in the second round.
\end{remark}

\section{Concepts From the Monomial Matrix Formalism}
\label{sec:monomial}

In this subsection we introduce some definitions and theorems from~\cite{Vandennest2011monomial}, and explain how they are connected to this work.

In~\cite{Vandennest2011monomial}, we consider a group $G=\langle U_1,\dots, U_m\rangle$, where each $U_j$ is a unitary monomial operator, i.e.
\begin{equation}
    U_j
    =P_j
     D_j
\end{equation}
where $P_j$ is a permutation matrix and $D_j$ is a diagonal unitary matrix.
Define $\PermutationGroup$ to be the permutation group generated by $P_j$.
The goal of~\cite{Vandennest2011monomial} is to study the space of states that satisfy
\begin{equation}
    U_j\mkern2mu
    \ket{\psi}
    =\ket{\psi}
    \qquad
    \text{for every $j=1,\dots,m$}.
\end{equation}
Given a computational basis state $\ket{x}$, following~\cite{Vandennest2011monomial} we define the orbit $\Orbit_x$ to be
\begin{equation}
    \Orbit_x
    =\setbuilder{\ket{y}}
                {\exists
                 P\in
                 \PermutationGroup\colon
                 P\mkern2mu
                 \ket{x}
                 =\ket{y}}.
\end{equation}
We also define $G_x$ to be the subgroup of all $U\in G$ that have~$\ket{x}$ as an eigenvector.
Then we have the following theorem
\begin{theorem}
    \label{thm:monomial}
    Consider a group $G=\langle U_1,\dots, U_m\rangle$ of monomial unitary matrices.
    
    (a) There exists a state $\ket{\psi}\neq 0$ stabilized by $G$ if and only if there exists a computational basis state $|x\rangle$ such that
    \begin{equation}
        \label{eq:condition_support}
        U\mkern2mu
        \ket{x}
        =\ket{x}
        \qquad
        \text{for all $U\in G_x$}.
    \end{equation}
    
    (b) For every computational basis state $|x\rangle$ satisfying~\eqref{eq:condition_support}, there exists a state $\ket{\psi_x}$ stabilized by $G$, which is of the form
    \begin{equation}
        \ket{\psi_x}
        =\frac{1}
              {\sqrt{\abs{\Orbit_x}}}
         \sum_{\ket{y}\in
               \Orbit_x}
         f(y)\mkern2mu
         \ket{y},
    \end{equation}
    where $|f(y)|=1$ for all $y\in \Orbit_x$.
    Moreover, there exists a subset $\{|x_1\rangle, \dots, |x_d\rangle\}$ (each satisfying~\eqref{eq:condition_support}) such that
    \begin{itemize}
        \item the orbits $\Orbit_{x_i}$ are mutually disjoint;
        \item the set of all $x$ satisfying~\eqref{eq:condition_support} is precisely $\Orbit_{x_1}\cup\dots\cup \Orbit_{x_d}$;
        \item $\set{\ket{\psi_{x_1}},\dots,\ket{\psi_{x_d}}}$ is a basis of the space stabilized by $G$.
            In particular, $d$ is the dimension of this space.
    \end{itemize}
\end{theorem}

\section{Computational Complexity of the XS-Stabilizer Problem}
\label{sec:complexity}

Here we address the computational complexity of determining whether a subgroup~$G$ of the Pauli-S group, specified in terms of a generating set, \emph{is} an XS-stabilizer group, i.e. whether there exists a quantum state~$\ket{\psi}\neq 0$ that is stabilized by~$G$.
More precisely, the problem can be formulated as
\begin{description}
    \item[Problem]
        \Problem{XS-Stabilizer}.
    \item[Input]
        A list of $s_j\in\set{0,\dots,7}$, $\vec{a}_j\in\set{0,1}^n$ and $\vec{b}_j\in\set{0,\dots,3}^n$ where $j=1,\dots,m$, which describe a set $\Generators=\{g_1, \dots, g_m\}\subset\PauliSGroupMany$.
    \item[Output]
        If there exists a quantum state~$\ket{\psi}\neq 0$ such that $g_j\mkern2mu\ket{\psi}=\ket{\psi}$ for every~$j$ then output \textsc{YES}; otherwise output \textsc{NO}.
\end{description}

We have the following theorem:

\begin{theorem}
    The \Problem{XS-Stabilizer} problem is \ComplexityClass{NP}-hard.
\end{theorem}
\begin{proof}
    We will show this via a reduction from the \Problem{Positive 1-in-3-Sat} problem which is \ComplexityClass{NP}-complete~\cite{Schaefer}.
    The \Problem{Positive 1-in-3-Sat} problem is to determine whether a set of logical clauses in $n$~Boolean variables can be satisfied simultaneously or not.
    Each clause has three variables exactly one of which must be satisfied.
    We may express such a clause~$C_j$ as
    \begin{equation}
        \label{eq:1in3sat}
        x_{w_{j
              1}}+
        x_{w_{j
              2}}+
        x_{w_{j
              3}}
        =1
    \end{equation}
    for variables~$x_{w_{jk}}\in\{0,1\}$ and $1\leq w_{jk}\leq n$.
    
    We construct a corresponding instance of the \Problem{XS-Stabilizer} problem by encoding each clause~$C_j$ in a generator
    \begin{equation}
        g_j
        =\mathrm{i}^3
         S_{w_{j
               1}}
         S_{w_{j
               2}}
         S_{w_{j
               3}}\in
         \PauliSGroupMany.
    \end{equation}
    Since all~$g_j$ are diagonal this stabilizer problem is equivalent to determining whether there exists a computational basis state~$\ket{x_1,\dots,x_n}$ stabilized by every~$g_j$, i.e. whether all equations
    \begin{equation}
        \label{eq:1in3sat_stabilizer}
        \mathrm{i}^{3+
                    x_{w_{j
                          1}}+
                    x_{w_{j
                          2}}+
                    x_{w_{j
                          3}}}
        =1
    \end{equation}
    have a common solution.
    Since~\eqref{eq:1in3sat} and~\eqref{eq:1in3sat_stabilizer} are equivalent we have shown that the \Problem{XS-Stabilizer} problem is at least as hard as the \Problem{Positive 1-in-3-Sat} problem.
\end{proof}

The \ComplexityClass{NP}-hardness of the \Problem{XS-Stabilizer} problem is in sharp contrast with the corresponding problem in the Pauli stabilizer formalism, which is known to be in~\ComplexityClass{P}.

Next we show that the \Problem{XS-Stabilizer} problem is in~\ComplexityClass{NP}, which means there is an efficient classical proof allowing to verify whether a group is indeed an XS-stabilizer group.
\begin{theorem}
    \label{thm:inside NP}
    The \Problem{XS-Stabilizer} problem is in \ComplexityClass{NP}.
\end{theorem}
\begin{proof}
    We first determine if $\Generators$ is an admissible generating set, which can be done efficiently.
    If it is not, $\gen{\Generators}$ cannot be an XS-stabilizer group, hence we output \textsc{NO}.
    If $\Generators$ is found to be an admissible generating set, we proceed with the group $G=\gen{\Generators}$ as follows.
    Given a computational basis state~$\ket{x}$, recall the definition of the set~$G_x$ in Section~\ref{sec:monomial}.
    Note that every XS-operator $\alpha^s X(\vec{a}\mkern1mu)\mkern2mu S(\vec{b}\mkern1mu)$ maps~$\ket{x}$ to $\lambda\mkern2mu\ket{x\oplus\vec{a}\mkern1mu}$ for some complex phase~$\lambda$.
    This implies that $G_x=\DiagonalSubgroup$ for every~$x$.
    Then, by Theorem~\ref{thm:monomial}(a), to check whether~$G$ is an XS-stabilizer group, we only need to check whether there is a computational basis state~$\ket{z}$ stabilized by~$\DiagonalSubgroup$.
    Note that a generating set $\set{D_1,\dots,D_r}$ of~$\DiagonalSubgroup$ can be computed efficiently owing to Lemma~\ref{lem:generators of diagonal subgroup}.
    Furthermore, $\ket{z}$ is stabilized by~$\DiagonalSubgroup$ if and only if it is stabilized by every generator~$D_j$.
    Summarizing, we find that $G$ is an XS-stabilizer group iff there exists a computational basis state~$\ket{z}$ satisfying $D_j\mkern2mu\ket{z}=\ket{z}$ for all~$j$.
    So if $G$ is an XS-stabilizer group, a classical string~$z$ satisfying these conditions will serve as a proof since the equations $D_j\mkern2mu\ket{z}=\ket{z}$ can be verified efficiently.
    This shows that the \Problem{XS-Stabilizer} problem is in~\ComplexityClass{NP}.
\end{proof}

\begin{corollary}
    \label{cor:NP-complete}
    The \Problem{XS-Stabilizer} problem is \ComplexityClass{NP}-complete.
\end{corollary}

\section{Regular XS-Stabilizer Groups}
\label{sec:regular}

The results in section \ref{sec:complexity} imply that working with the XS-stabilizer formalism is computationally hard in general.
In order to recover tractability, we may impose certain restrictions on the type of stabilizers we can have.
In particular, we have the following theorem.

\begin{theorem}
    \label{thm:efficient result}
    Let $\Generators=\set{g_1,\dots,g_m}\subset\PauliSGroupMany$ and $G=\gen{\Generators}$.
    If $\DiagonalSubgroup=\ZSubgroup$ then the \Problem{XS-Stabilizer} problem is in~\ComplexityClass{P}.
    Moreover, this condition can be checked efficiently.
\end{theorem}
\begin{proof}
    As in the proof of Theorem~\ref{thm:inside NP}, we only need to consider admissible generating sets, and we need to decide whether there exists a computational basis state stabilized by~$\DiagonalSubgroup$.
    Since Lemma~\ref{lem:generators of diagonal subgroup} states that~$\DiagonalSubgroup$ is generated by $\set{D_1,\dots,D_r}=\Commutators\cup\Squares\cup\set{g(\vec{u}_i)}$ with $r=\poly(m,n)$ we can efficiently check whether $\set{g(\vec{u}_i)}\subset\ZSubgroup$ or some~$g(\vec{u}_i)$ contains an $S$~operator, i.e. whether $\DiagonalSubgroup=\ZSubgroup$ or not.
    Now let us assume that $\DiagonalSubgroup=\ZSubgroup$.
    By Proposition~\ref{prop:stabilizer_group} we immediately deduce that $D_j=(-1)^{s_j}Z(\vec{c}_j)$ for some $s_j\in\set{0,1}$ and $\vec{c}_j\in\set{0,1}^n$.
    Then a computational basis state~$\ket{\vec{z}\mkern2mu}$ stabilized by~$\DiagonalSubgroup$ is equivalent to a nontrivial common solution of the equations $\vec{c}_j\cdot\vec{z}=s_j$.
    These are polynomially many linear equations in $n$~variables over~$\mathbb{Z}_2$ and can hence be solved in polynomial time.
    This proves that the restricted \Problem{XS-Stabilizer} problem is indeed in~\ComplexityClass{P}.
\end{proof}

Motivated by this result, we call any XS-stabilizer group~$G$ with $\DiagonalSubgroup=\ZSubgroup$ a \emph{regular XS-stabilizer group}.

\begin{remark}
    We want to mention that diagonal elements containing $S$~operators play a crucial role in certain examples, which we will discuss in section~\ref{sec:non regular xs stabilizer}.
    However, as will be proved in theorem~\ref{thm:non-regular state has regular group}, we can construct a basis of the space stabilized by a general XS-stabilizer group such that each basis state is described by a regular XS-stabilizer group.
\end{remark}

Here we give a sufficient condition for an XS-stabilizer group to be regular.

\begin{lemma}
    \label{lemma:sufficient condition for regular group}
    If an XS-stabilizer group $G\subset\PauliSGroupMany$ has a generating set $\Generators=\set{g_1,\dots,g_t,g_{t+1},\dots,g_m}$ where
    \begin{equation}
        \label{eq:standard form when no S}
        g_j
        =\begin{cases}
             \alpha^{s_j}
             X(\vec{a}_j)\mkern2mu
             S(\vec{b}_j)
             & \text{if $j\leq t$}, \\
             (-1)^{s_j}
             Z(\vec{c}_j)
             & \text{else}
         \end{cases}
    \end{equation}
    and $\set{\vec{a}_1,\dots,\vec{a}_t}$ are linearly independent over $\mathbb{Z}_2$, then it is regular.
\end{lemma}
\begin{proof}
    We see from Lemma~\ref{lem:relative_standard_form} that $\DiagonalSubgroup$ is generated by $\Commutators$, $\Squares$ and those elements~$g(\vec{x}\mkern1mu)=g_1^{x_1}\cdots g_m^{x_m}$ which are diagonal.
    In order to show that $\DiagonalSubgroup=\ZSubgroup$ we only need to show that all diagonal elements~$g(\vec{x}\mkern1mu)$ are Z-type operators.
    It follows from Lemma~\ref{lem:multiplication} that $g(\vec{x}\mkern1mu)\propto X(A\vec{x}\mkern1mu)\mkern2mu S(\vec{b}')$ for the $n\times m$ matrix
    \begin{equation}
        A
        =[\begin{matrix}
              \vec{a}_1 & \dots & \vec{a}_t & 0 & \dots & 0
          \end{matrix}]
    \end{equation}
    and some~$\vec{b}'$.
    This implies that $g(\vec{x}\mkern1mu)$ is a diagonal operator if and only if $A\vec{x}=0$ over~$\mathbb{Z}_2$.
    Since the bit strings~$\vec{a}_j$ are linearly independent this can only be true if $x_j=0$ for all $j\leq t$.
    Thus every diagonal element $g(\vec{x}\mkern1mu)=g_{t+1}^{x_{t+1}}\cdots g_m^{x_m}$ is indeed a Z-type operator.
\end{proof}

\begin{theorem}[Normal form]
    \label{thm:normal_form}
    Every XS-stabilizer group $G\subset\PauliSGroupMany$ has a generating set $\Generators=\set{g_1,\dots,g_t,g_{t+1},\dots,g_m}$ where
    \begin{equation}
        \label{eq:normal_form}
        g_j
        =\begin{cases}
             \alpha^{s_j}
             X(\vec{e}_j,
               \vec{w}_j)\mkern2mu
             S(\vec{b}_j)
             & \text{if $j\leq t$}, \\
             \mathrm{i}^{s_j}
             S(\vec{b}_j)
             & \text{else}
         \end{cases}
    \end{equation}
    for the canonical basis vectors $\vec{e}_j\in\mathbb{Z}_2^t$ and some $\vec{w}_j\in\mathbb{Z}_2^{n-t}$.
    Furthermore, $\DiagonalSubgroup=\gen{g_{t+1},\dots,g_m}$.
    
    Any other generating set $\Generators'=\set{h_1,\dots,h_l}$ of~$G$ can be reduced to~$\Generators$ in $\poly(n,l)$ time and the length of~$\Generators$ is $\poly(n,l)$.
\end{theorem}
\begin{proof}
    We will prove this by explicitly reducing~$\Generators'=\set{h_1,\dots,h_l}$.
    By Lemma~\ref{lem:generators of diagonal subgroup} we can efficiently find a generating set $\set{D_1,\dots,D_r}$ of~$\DiagonalSubgroup$ from~$\Generators'$.
    Next we can extend it to a generating set of~$G$ by adding a minimal subset of~$\Generators'$ which can be found efficiently.
    Indeed, suppose we added all generators $h_j=\alpha^{s_j}X(\vec{a}_j)\mkern2mu S(\vec{b}_j)$.
    If $\set{\vec{a}_1,\dots,\vec{a}_l}$ are linearly dependent over~$\mathbb{Z}_2$, we may assume that $\vec{a}_l=y_1\vec{a}_1\oplus\dots\oplus y_{l-1}\vec{a}_{l-1}$ for some $y_j\in\mathbb{Z}_2$ without loss of generality.
    Then Lemma~\ref{lem:multiplication} implies $h_1^{y_1}\cdots h_{l-1}^{y_{l-1}}=D h_l$ for some diagonal element~$D\subset\PauliSGroupMany$.
    It is not difficult to see that actually $D\in G$ and hence $D\in\DiagonalSubgroup$.
    This means that the generator $h_l\in\gen{D_1,\dots,D_r,h_1,\dots,h_{l-1}}$ is redundant.
    So we can find a minimal subset of~$\Generators'$ by finding bit strings~$\vec{a}_j$ which form a basis of the linear space $\gen{\vec{a}_1,\dots,\vec{a}_l}$.
    Note that this only involves Gaussian elimination over~$\mathbb{Z}_2$.
    Now let $\set{h_1,\dots,h_t}$ be the desired subset of~$\Generators'$ after relabeling the generators~$h_j$, so that $G=\gen{h_1,\dots,h_t,D_1,\dots,D_r}$.
        
    We can arrange the bit strings~$\vec{a}_1$, …, $\vec{a}_t$ in the $n\times t$ matrix
    \begin{equation}
        A
        \coloneqq
         [\begin{matrix}
              \vec{a}_1 & \dots & \vec{a}_t
          \end{matrix}].
    \end{equation}
    Since the~$\vec{a}_j$ are linearly independent, by Gaussian elimination and suitable permutation of columns, we can efficiently transform~$A$ into
    \begin{equation}
        P
        A
        R
        =\begin{bmatrix}
             \IdentityMatrix_t \\
             W
         \end{bmatrix}
    \end{equation}
    for some $n\times n$ permutation matrix~$P$, $t\times t$ invertible matrix~$R$ and $(n-t)\times t$ matrix~$W$.
    By relabeling the qubits according to the permutation defined by $P$ and by multiplying the generators $h_1$, …, $h_t$ according to the transformation $R$, we obtain an equivalent set of generators $\set{g_1,\dots,g_t}$ such that $g_j=\alpha^{s_j'} X(\vec{e}_j,\vec{w}_j)\mkern2mu S(\vec{b}_j')$ for some $s_j'$ and $\vec{b}_j'$.
    Here $\vec{w}_j$ denotes the $j$-th column of~$W$.
    
    Finally note that $D_j=\mathrm{i}^{s_j'}S(\vec{b}_j')$ for some~$s_j'$ and~$\vec{b}_j'$ since~$G$ is an XS-stabilizer group.
    We conclude that $\Generators=\set{g_1,\dots,g_t,D_1,\dots,D_r}$ is the desired generating set of~$G$.    
\end{proof}

\begin{remark}
    The converse does not hold since there may not exist any state $\ket{\psi}\neq 0$ which is stabilized by generators of the form~\eqref{eq:normal_form}.
    Deciding this is an \ComplexityClass{NP}-complete problem by Corollary~\ref{cor:NP-complete}.
\end{remark}

\begin{corollary}
    \label{cor:normal_form_regular}
    Every regular XS-stabilizer group $G\subset\PauliSGroupMany$ has a generating set $\Generators=\set{g_1,\dots,g_t,g_{t+1},\dots,g_m}$ where
    \begin{equation}
        g_j
        =\begin{cases}
             \alpha^{s_j}
             X(\vec{e}_j,
               \vec{w}_j)\mkern2mu
             S(\vec{b}_j)
             & \text{if $j\leq t$}, \\
             (-1)^{s_j}
             Z(\vec{c}_j)
             & \text{else}
         \end{cases}
    \end{equation}
    for the canonical basis vectors $\vec{e}_j\in\mathbb{Z}_2^t$ and some $\vec{w}_j\in\mathbb{Z}_2^{n-t}$.
    Furthermore, $\DiagonalSubgroup=\gen{g_{t+1},\dots,g_m}$.
\end{corollary}

Unless stated otherwise, we will always work with regular XS-stabilizer groups from this point on.

\section{Constructing a Basis of a Regular XS-Stabilizer Code}
\label{sec:general phase}

The goal of this section is to construct a basis for a regular XS-stabilizer code.
For each state of the basis, we will give an explicit form of its expansion in the computational basis.
To achieve this (in sections \ref{sec:constructing_basis}, \ref{sec:logical_operators} and \ref{sec:stronger characterization}), we will first introduce some preliminary material on quadratic and cubic functions (in section \ref{sec:cubic}).

\subsection{Quadratic and Cubic Functions}
\label{sec:cubic}

In this paper we will often deal with functions of the form $i^{\sum x_j x_k}$ or $(-1)^{\sum x_j x_k x_l}$ with $x_j\in\{0, 1\}$.
Here we list some properties of such functions that will become useful later.
First we have the following lemma.

\begin{lemma}[Exponentials of parities]
    \label{lem:parity to high order}
    Let $x_1$, …, $x_n\in\{0,1\}$.
    Then
    \begin{align}
        \alpha^{x_1\oplus
                \dots\oplus
                x_n}
        & =\alpha^{\sum_j
                   x_j}\mkern2mu
           \mathrm{i}^{-\mkern-3mu
                       \sum_{j
                             <k}
                       x_j
                       x_k}\mkern2mu
           (-1)^{\sum_{j
                       <k
                       <l}
                 x_j
                 x_k
                 x_l} \\
    \intertext{and}
        i^{x_1\oplus
           \dots\oplus
           x_n}
        & =\mathrm{i}^{\sum_j
                       x_j}\mkern2mu
           (-1)^{\sum_{j
                       <k}
                 x_j
                 x_k}.
    \end{align}
\end{lemma}
\begin{proof}
    The first equation is proved in~\cite{BravyiHaah}.
    The second equation follows by similar arguments.
\end{proof}

All exponents in this Lemma are homogeneous polynomials of degree at most~$3$.
Below we will often only be interested in whether a given exponent is a quadratic or cubic polynomial, but not in its concrete form.
We will therefore use $l(x)$, $q(x)$ and $c(x)$ to represent arbitrary linear, quadratic and cubic polynomials in~$\mathbb{Z}[x_1,\dots,x_n]$ respectively.
(These polynomials need not be homogeneous.)
Using this notation, the Lemma can be summarized as
\begin{align}
    \alpha^{x_1\oplus
            \dots\oplus
            x_n}
    & =\alpha^{l(x)}\mkern2mu
       \mathrm{i}^{q(x)}\mkern2mu
       (-1)^{c(x)}, \\
    \mathrm{i}^{x_1\oplus
                \dots\oplus
                x_n}
    & =i^{l(x)}\mkern2mu
       (-1)^{q(x)}.
\end{align}

Let $\CovariantFunctions$ denote the class of all functions $f\colon\{0, 1\}^n\to \mathbb{C}$ having the form
\begin{equation}
    f(x)
    =\alpha^{l(x)}\mkern2mu
     \mathrm{i}^{q(x)}\mkern2mu
     (-1)^{c(x)}
    \qquad
    \text{for all $x\in\set{0,1}^n$}.
\end{equation}
Note that $\CovariantFunctions$ is closed under multiplication, i.e. if $f$ and $g$ belong to this class, then so does $fg$.

\begin{remark}
    \label{rem:covariant_functions_generators}
    Linear phases~$\alpha^{l(x)}$ are generated by $\{\alpha^{x_j}\}$ via multiplication since $\alpha^{l(x)}=\prod_{j=1}^n (\alpha^{x_j})^{\lambda_j}$ for a linear polynomial $l(x)=\sum_{j=1}^n \lambda_j x_j$.
    By the same token, $\{\mathrm{i}^{x_j},\mathrm{i}^{x_j x_k}\}$ generate all quadratic phases~$\mathrm{i}^{q(x)}$ and $\{(-1)^{x_j},(-1)^{x_j x_k},(-1)^{x_j x_k x_l}\}$ all cubic phases~$(-1)^{c(x)}$.
    This implies
    \begin{equation}
        \CovariantFunctions
        =\gen{\alpha^{x_j},
              \mathrm{i}^{x_j
                          x_k},
              (-1)^{x_j
                    x_k
                    x_l}}.
    \end{equation}
    \qed
\end{remark}

\begin{lemma}[Covariance]
    Let $Q$ be a linear map in the vector space~$\mathbb{Z}_2^n=\{0,1\}^n$.
    If a function~$f$ belongs to~$\CovariantFunctions$ then so does $f\circ Q$.
\end{lemma}
\begin{proof}
    Let $x=(x_1,\dots,x_n)\in\mathbb{Z}_2^n$ and $y=Q(x)$, so each~$y_j=\bigoplus_{k=1}^n Q_{jk}x_k$ is the parity of some substring of~$x$.
    By Remark~\ref{rem:covariant_functions_generators} it is enough to show that the phases~$\alpha^{y_j}$, $\mathrm{i}^{y_j y_k}$ and $(-1)^{y_j y_k y_l}$ belong to~$\CovariantFunctions$ since any~$f(y)$ equals some product of them.
    
    It is immediate from Lemma~\ref{lem:parity to high order} that every $\alpha^{y_j}\in\CovariantFunctions$.
    Furthermore, Lemma~\ref{lem:parity to high order} implies
    \begin{equation}
        \mathrm{i}^{y_j
                    y_k}
        =\bigl(\mathrm{i}^{l(x)}\mkern2mu
               (-1)^{q(x)}\bigr)^{y_k}
        =\mathrm{i}^{l(x)\mkern2mu
                     y_k}\mkern2mu
         (-1)^{q(x)\mkern2mu
               y_k}
        =(\mathrm{i}^{y_k})^{l(x)}\mkern2mu
         (-1)^{c(x)}
    \end{equation}
    for some linear, quadratic and cubic polynomials~$l$, $q$ and~$c$ respectively.
    Invoking Lemma~\ref{lem:parity to high order} once more we see that $(\mathrm{i}^{y_k})^{l(x)}\in\CovariantFunctions$, thus $\mathrm{i}^{y_j y_k}\in\CovariantFunctions$.
    Finally, it is easy to check that $(-1)^{y_j y_k y_l}=(-1)^{c'(x)}$ for some cubic polynomial~$c'$, which shows $(-1)^{y_j y_k y_l}\in\CovariantFunctions$.
\end{proof}

\subsection{Constructing a Basis}
\label{sec:constructing_basis}

Consider an $n$-qubit XS-stabilizer code $\RegularXSCode$ with regular stabilizer group $G=\langle g_1, \dots, g_m\rangle$.
Without loss of generality, we may assume that the generators $g_j$ have the form given in Corollary~\ref{cor:normal_form_regular}.
We will construct a basis for this XS-stabilizer code by applying the monomial matrix method outlined in section \ref{sec:monomial}.

Denote the $(n-t)\times t$ matrix $W\coloneqq[\vec{c}_1|\cdots|\vec{c}_t]$ (as in the proof of Theorem~\ref{thm:normal_form}).
Define $\OrbitSpace$ to be the linear subspace of $\mathbb{Z}_2^n$ consisting of all couples $(x, Wx)$ with $x\in\mathbb{Z}_2^t$.
A basis of this space is given by the vectors $(\vec{e}_i, \vec{c}_i)$ where $\vec{e}_i$ is the $i$-th canonical basis vector in $\mathbb{Z}_2^t$.
Furthermore consider the set~$\InvariantSetDiagonal$ of those $n$-bit strings $z$ satisfying $D|z\rangle = |z\rangle$ for all $D\in\DiagonalSubgroup$.
This coincides with the set of all $z$ satisfying $g_j|z\rangle = |z\rangle$ for all $j=t+1,\dots,m$, since these $g_j$ generate the diagonal subgroup by Corollary~\ref{cor:normal_form_regular}.
Since each of these $g_j$ has the form $(-1)^{s_j} Z(\vec{b}_j)$, $\InvariantSetDiagonal$ is the set of all $z$ satisfying $\vec{b}_j^Tz = s_j$ for all $j=t+1,\dots,m$.
The set~$\InvariantSetDiagonal$ is thus an affine subspace of $\mathbb{Z}_2^n$.
A basis of~$\InvariantSetDiagonal$ can be computed efficiently.

Note that every Pauli-S operator $g= i^{s} X(\vec{a})S(\vec{b})$ is a monomial unitary matrix, where $X(\vec{a})$ is the corresponding permutation matrix and $i^{s} S(\vec{b})$ the corresponding diagonal matrix (recall subsection~\ref{sec:monomial}).
It follows that the permutation group $\PermutationGroup$ associated with $G$ is generated by the operators $X_j X(\vec{c}_j)$.
Recalling the definition of the space~$\OrbitSpace$, this implies that $\PermutationGroup= \{X(v)\mid v\in\OrbitSpace\}$.
Furthermore, the orbit of a computational basis state $|x\rangle$ is the coset of~$\OrbitSpace$ containing $x$ i.e. $\Orbit_x= x+\OrbitSpace$.
To see this, note that
\begin{equation}
    \setbuilder{X(v)\mkern2mu
                \ket{x}}
               {v\in
                \OrbitSpace}
    =\setbuilder{\ket{x+
                      v}}
                {v\in
                 \OrbitSpace}.
\end{equation}
Applying theorem \ref{thm:monomial}(b), we conclude that there exist orbits $\Orbit_1$, …, $\Orbit_d$ such that
\begin{equation}
    \InvariantSetDiagonal
    =\Orbit_1\cup
     \dots\cup
     \Orbit_d
\end{equation}
and $d$ coincides with the dimension of the XS-stabilizer code.
Note that we can efficiently compute $d$: each orbit has size~$|\OrbitSpace|$ and thus $d|\OrbitSpace|= |\InvariantSetDiagonal|$; since both~$|\OrbitSpace|$ and~$|\InvariantSetDiagonal|$ can be computed efficiently (as we know bases for both of these spaces), we can efficiently compute $d$.
Note that $d$ is a power of two, since both~$|\OrbitSpace|$ and~$|\InvariantSetDiagonal|$ are powers of two.
Finally, a set of strings $\vec{\lambda}_1,\dots,\vec{\lambda}_d\in\InvariantSetDiagonal$ such that $\Orbit_i = \vec{\lambda}_i + \OrbitSpace$ can be computed in $\poly(n,m,d)$ time.

As a corollary of the above discussion, we also note:

\begin{lemma}
    \label{thm:dimension}
    The XS-stabilizer code $\RegularXSCode$ is one-dimensional, i.e. it is an XS-stabilizer state, iff $|\OrbitSpace|= |\InvariantSetDiagonal|$.
\end{lemma}

By theorem \ref{thm:monomial}, for each vector $\vec{\lambda}\in\InvariantSetDiagonal$ there exists a state
\begin{equation}
    \label{eq:stab_state}
    \ket{\psi}
    =\smashoperator[l]{\sum_{x\in
                             \mathbb{Z}_2^t}}
     g(x)\mkern2mu
     \ket{x+
          \vec{\lambda}_1,
          W
          x+
          \vec{\lambda}_2}.
\end{equation}
stabilized by $G$, where $g(x)$ is some function that satisfies $|g(x)|=\abs{\OrbitSpace}^{-\frac{1}{2}}$ for all $x$ and where we have partitioned $\vec{\lambda}=(\vec{\lambda}_1, \vec{\lambda}_2)$ with $\vec{\lambda}_1$ representing the first $t$ components and $\vec{\lambda}_2$ the last $n-t$ components of $\lambda$.
By performing the substitution $x\mapsto x+\vec{\lambda}_1$ and denoting $f(x)\coloneqq g(x+\vec{\lambda}_1)$ and $\vec{\mu}\coloneqq\vec{\lambda}_2 + W\vec{\lambda}_1$, we find
\begin{equation}
    \label{eq:the W matrix}
    \ket{\psi}
    =\smashoperator[l]{\sum_{x\in
                             \mathbb{Z}_2^t}}
     f(x)\mkern2mu
     \ket{x,
          W
          x+
          \vec{\mu}}.
\end{equation}

\begin{lemma}
    \label{thm:compute_form_psi}
    Suppose the state $|\psi\rangle$ of the form~\eqref{eq:stab_state} is stabilized by a regular XS-stabilizer group $G$.
    Then based on the generators $g_j$ of $G$ and the string $\vec{\lambda}$, the following data in equation~\eqref{eq:the W matrix} can be computed efficiently:  (i) the matrix $W$ and the string $\vec{\mu}$; (ii) the function $f(x)$.
    Moreover, we can efficiently find a complete set of stabilizers that uniquely stabilize $|\psi\rangle$.
\end{lemma}
\begin{proof}
    That the matrix $W$ and the string $\vec{\mu}$ can be computed efficiently was shown in the argument above the lemma.
    In order to prove that the function~$f(x)$ can be computed efficiently, we note that $g_j\mkern2mu\ket{\psi} =\ket{\psi}$ for all $1\leq j\leq t$ and thus
    \begin{equation}
        g_t^{x_t}
        \cdots
        g_1^{x_1}
        \ket{\psi}
        =\ket{\psi}.
    \end{equation}
    Comparing the terms with $\ket{x}$ on both sides, we have
    \begin{equation}
        \label{eq:calculating phases by applying g_first}
        g_t^{x_t}
        \cdots
        g_1^{x_1}
        \ket{0,
             \vec{\mu}}
        =f(x)\mkern2mu
         \ket{x,
              W
              x+
              \vec{\mu}}.
    \end{equation}
    This shows that $f(x)$ can be computed by computing the phase which appears when computing $g_t^{x_t}\cdots g_1^{x_1}\ket{0,\vec{\mu}}$.
    Finally, we construct a complete stabilizer of $|\psi\rangle$, showing that this state is an XS-stabilizer state.
    
    We work with the representation \eqref{eq:stab_state} of $|\psi\rangle$, which implies that $|\psi\rangle$ has the form
    \begin{equation}
        \label{eq:psi_support_coset}
        \ket{\psi}
        =\sum_{v\in
               \OrbitSpace+
               \vec{\lambda}}
         \alpha_v
         \ket{v}
    \end{equation}
    (for some coefficients $\alpha_v$) where we recall the definition of $\OrbitSpace\subset\mathbb{Z}_2^n$, which is the linear subspace of all pairs $(x, Wx)$.
    The orthogonal complement of~$\OrbitSpace$ is the $(n-t)$-dimensional space of all pairs $(W^Ty, y)$ with $y\in \mathbb{Z}_2^{n-t}$.
    A basis $\{\vec{z}_1, \dots, \vec{z}_{n-t}\}$ of the latter space can be computed efficiently.
    It follows that any string $v\in \mathbb{Z}_2^n$ belongs to $\OrbitSpace + \vec{\lambda}$ if and only if $\vec{z}_j^T v = \vec{z}_j^T \vec{\lambda}$.
    Define operators
    \begin{equation}
        \label{eq:D_j}
        D_j
        \coloneqq
         (-1)^{\vec{z}_j\mkern1mu^T
               \vec{\lambda}}\mkern2mu
         Z(\vec{z}_j)
        \qquad
        j
        =1,
         \dots,
         n-
         t.
    \end{equation}
    Then for any string $v\in \mathbb{Z}_2^n$, we have
    \begin{equation}
        v\in
        \OrbitSpace+
        \vec{\lambda}
        \quad
        \Leftrightarrow
        \quad
        \text{$D_j\mkern2mu\ket{v}=\ket{v}$ for every $j=1,\dots,n-t$}.
    \end{equation}
    We now supplement the initial generators $\{g_1, \dots, g_m\}$ of $G$ with the operators $D_j$.
    Denote the resulting set by $\Generators'$ and let $G'$ denote the group generated by $\Generators'$.
    Clearly, $G'$ is regular since we supplemented the initial group $G$, which was regular, with Z-type operators.
    We now claim that $G'$ has $|\psi\rangle$ as its unique stabilized state.
    First, combining \eqref{eq:psi_support_coset} with \eqref{eq:condition_support} shows that $|\psi\rangle$ is stabilized by every operator $D_j$.
    This shows that $|\psi\rangle$ is stabilized by $G'$.
    Second, let $\DiagonalSubgroup'$ be the diagonal subgroup of $G'$.
    Let $U$ be the set of all $u\in \mathbb{Z}_2^n$ satisfying $D|u\rangle = |u\rangle$ for all $D\in\DiagonalSubgroup'$.
    According to the argument above lemma \ref{thm:dimension}, $U$ is the disjoint union of cosets of~$\OrbitSpace$; the number of such cosets is the dimension of the code stabilized by $G'$.
    We show that in fact $U=\OrbitSpace+\vec{\lambda}$, implying that this dimension is one, so that $|\psi\rangle$ is the unique stabilized state.
    Since each $D_j$ belongs to~$\DiagonalSubgroup'$, every $u\in U$ must satisfy $D_j|u\rangle = |u\rangle$ for all $j=1, \dots, n-t$.
    With \eqref{eq:condition_support} this shows that $U\subseteq\OrbitSpace+\vec{\lambda}$.
    Furthermore, since $D|\psi\rangle= |\psi\rangle$ for every $D\in\DiagonalSubgroup'$ and since $|\psi\rangle$ has the form \eqref{eq:psi_support_coset}, it follows from that every $u\in\OrbitSpace+\vec{\lambda}$ satisfies $D|u\rangle = |u\rangle$ for all $D\in\DiagonalSubgroup'$.
    This shows that $\OrbitSpace + \vec{\lambda}\subseteq U$ and thus $\OrbitSpace + \vec{\lambda}= U$.
\end{proof}

Next we determine the form of the function~$f$ in more detail.
This will be e.g. useful for constructing a quantum circuit to generate $|\psi\rangle$ (see section \ref{sec:algorithms}).
For simplicity of notation, we rescale the state $|\psi\rangle$ by multiplying it with a suitable constant so that we can assume $f(0)=1$.
We will show that then $f$ belongs to the class $\CovariantFunctions$ introduced in section \ref{sec:cubic}.
We recall the identity~\eqref{eq:calculating phases by applying g_first}.
To show that $f\in \CovariantFunctions$, we will compute the left hand side of the equation, and we will do this by using induction on $k$.
As for the trivial step of the induction, if $k=1$, it is easy to see that $f\in \CovariantFunctions$.
Set $x[k]\coloneqq(x_1,\dots,x_k,0,\dots,0)$, and assume that
\begin{equation}
    g_k^{x_k}
    \cdots
    g_1^{x_1}\mkern2mu
    \ket{0,
         \vec{\mu}}
    =f(x[k])\mkern2mu
     \ket{x[k],
          W
          x[k]+
          \vec{\mu}}
\end{equation}
with $f\in \CovariantFunctions$.
If we apply an $S$ operator to the $l$-th qubit of $|x[k]\rangle$, we simply obtain the phase $i^{x_l}$.
Second, the $j$-th bit of $W x[k]+\vec{\mu}$ is
\begin{equation}
    W_{j
       1}
    x_1\oplus
    \dots\oplus
    W_{j
       k}
    x_k\oplus
    \mu_j.
\end{equation}
Therefore, if we apply an $S$ operator on the corresponding qubit, we will obtain the phase
\begin{equation}
    \mathrm{i}^{W_{j
                   1}
                x_1\oplus
                \dots\oplus
                W_{j
                   k}
                x_k\oplus
                \mu_j}
    =\mathrm{i}^{W_{j
                    1}
                 x_1+
                 \dots+
                 W_{j
                    k}
                 x_k+
                 \mu_j}\mkern2mu
     (-1)^{q(x[k])},
\end{equation}
for some quadratic polynomial $q(x[k])$ (which contains all possible products of $W_{jl}x_l$ and $\lambda_{j}$, where we have applied lemma \ref{lem:parity to high order}.
It is then easy to check that
\begin{equation}
    g_{k+
       1}^{x_{k+
              1}}\mkern2mu
    \ket{x[k],
         W
         x[k]+
         \vec{\mu}}
    =\alpha^{l(x[k+
                 1])}\mkern2mu
     \mathrm{i}^{q(x[k+
                     1])}\mkern2mu
     (-1)^{c(x[k+
               1])}\mkern2mu
     \ket{x[k+
            1],
          W
          x[k+
            1]+
          \vec{\mu}},
\end{equation}
where $l$, $q$ and $c$ are linear, quadratic and cubic polynomials in~$x_1$, …, $x_{k+1}$ respectively.
So we can conclude that
\begin{equation}
    \label{eq:calculating phases by applying g}
    g_t^{x_{k+
            1}}
    \cdots
    g_1^{x_1}\mkern2mu
    \ket{0,
         \vec{\mu}}
    =f(x[k+
         1])\mkern2mu
     \ket{x[k+
            1],
          W
          x[k+
            1]+
          \vec{\mu}}
\end{equation}
for some $f\in \CovariantFunctions$.

In the argument above, for any $j$ we can also check that in the phase $f(x)= \alpha^{l(x)}i^{q(x)} (-1)^{c(x)}$, the part that depends on both $\mu_{j}$ and $x$ is
\begin{equation}
    \label{eq:coefficient depend on lambda}
    \mathrm{i}^{l(x)\mkern2mu
                \mu_j}\mkern2mu
    (-1)^{q(x)\mkern2mu
          \mu_j}
\end{equation}
This will be useful for finding the logical operators in section \ref{sec:logical_operators}.

Summarizing, we have shown:

\begin{theorem}
    \label{thm:form_of_stab_states}
    Every regular XS-stabilizer state on $n$~qubits has the form
    \begin{equation}
        \label{eq:form_of_stab_states}
        \ket{\psi}
        =\frac{1}
              {\sqrt{2^t}}
         \sum_{x\in
               \mathbb{Z}_2^t}
         f(x)\mkern2mu
         \ket{x,
              W
              x+
              \vec{\mu}}
         \qquad
         \text{with $f(x)=\alpha^{l(x)}\mkern2mu\mathrm{i}^{q(x)}\mkern2mu(-1)^{c(x)}$}.
    \end{equation}
    up to a permutation of the qubits.
    Moreover, the polynomials~$l(x)$, $q(x)$ and~$c(x)$ can be computed efficiently.
\end{theorem}

Below in theorem~\ref{thm:non-regular state has regular group} we will show that every XS-stabilizer state affords a regular stabilizer group, so that the above theorem in fact applies to all XS-stabilizer states.

It is interesting to compare theorem \ref{thm:form_of_stab_states} with a similar result for Pauli stabilizer states \cite{Dehaene2003the}: every Pauli stabilizer state also has the form~\eqref{eq:form_of_stab_states}, but where
\begin{equation}
    \label{eq:form_of_pauli_stab_state}
    f(x)
    =\mathrm{i}^{l(x)}\mkern2mu
     (-1)^{q(x)},
\end{equation}
i.e. there are no cubic terms $(-1)^{c(x)}$, no quadratic terms $i^{q(x)}$ and no linear terms $\alpha^{l(x)}$.
It is also known that every state~\eqref{eq:form_of_stab_states} with $f$ having the form~\eqref{eq:form_of_pauli_stab_state} is a valid Pauli stabilizer state.
A similar statement does not hold for XS-stabilizer states i.e. not all functions $f\in\CovariantFunctions$ are allowed.
We will revisit this property in section \ref{sec:stronger characterization}.

\subsection{Logical Operators}
\label{sec:logical_operators}

Logical operators play a very important role in understanding the Pauli stabilizer formalism and performing fault tolerant quantum computation.
They represent the $X$ and $Z$ operators on the encoded qubits.
Since it is not clear yet how to define logical operators for an XS-stabilizer group, in this section we will just construct a set of operators that preserve the code space $\RegularXSCode$ (which we assume to have dimension~$2^s$) and obey the same commutation relations as $\{X_j, Z_j\}$ on $s$ qubits.

We already know there is a set of the basis of a regular XS-stabilizer code $\RegularXSCode$ which has the form
\begin{equation}
    \ket{\psi_j}
    =\smashoperator[l]{\sum_{x\in
                             \mathbb{Z}_2^t}}
     f_j(x)\mkern2mu
     \ket{x,
          W
          x+
          \vec{\mu}_j}.
\end{equation}
Note that $\vec{\mu}_j$ can be viewed as elements of the quotient space~$\InvariantSetDiagonal/\OrbitSpace$ (for definition of~$\InvariantSetDiagonal$ and~$\OrbitSpace$ see section~\ref{sec:constructing_basis}), which is an affine subspace.
So we know that up to some permutation of qubits, each $\vec{\mu}_j$ has the form
\begin{equation}
    (y_1,
     \dots,
     y_s,
     \Lambda
     y+
     \vec{\lambda}).
\end{equation}
Here $\Lambda$ and $\vec{\lambda}$ can be found by Gaussian elimination, and $s$ is the dimension of the space~$\InvariantSetDiagonal/\OrbitSpace$.
So for each $\mu_j$, we can find a corresponding $y$.
Assume we do not need to do permutation of qubits in the above step, the vector $\ket{x,Wx+\vec{\mu}_j}$ becomes
\begin{equation}
    \ket{x,
         W
         x+
         (y,
          \Lambda
          y+
          \vec{\lambda})}.
\end{equation}
For simplicity of notation, we will assume $\vec{\lambda}=0$.
The case $\vec{\lambda}\neq 0$ can be dealt with similarly by the following procedure.
We will show how to find operators $\bar{Z}_k$ and $\bar{X}_k$ for $k=1,\dots,s$, such that they act on the state $\ket{\psi(y)}\equiv\ket{\psi_j}$ by the following
\begin{align}
    \bar{Z}_k\mkern2mu
    \ket{\psi(y)}
    & =(-1)^{y_k}\mkern2mu
       \ket{\psi(y)} \\
    \bar{X}_k\mkern2mu
    \ket{\psi(y)}
    & =\ket{\psi(y+
                 \vec{e}_k)},
\end{align}
where $\vec{e}_k$ is the $k$-th canonical basis vector.
It is then easy to see within the space spanned by $\{\ket{\psi_j}\}$, we have $\bar{X}_j\bar{Z}_k=(-1)^{\delta(j,k)}\bar{Z}_k\bar{X}_j$, where $\delta(j,k)$ is the Kronecker delta function.
The $\bar{Z}_k$ can be found straightforwardly.
For example, we can construct $\bar{Z}_k$ by noticing
\begin{equation}
    (-1)^{y_k}
    =(-1)^{\sum_l
           W_{k
              l}
           x_l+
           y_k}
     \prod_l
     (-1)^{W_{k
              l}
           x_l},
\end{equation}
where $W_{kl}$ is the matrix element of $W$.
In the r.h.s of the above equation, the term $(-1)^{\sum_l W_{kl}x_l+y_k}$ can be achieved by applying $Z$ on the $k$-th qubit in the block $\ket{Wx+(y,\Lambda y)}$, and the terms $(-1)^{W_{kl}x_l}$ can be obtained by applying $Z$ on the corresponding qubits in the block $\ket{x}$.

The construction of $\bar{X}_k$ is more involved and different from the one used in the Pauli stabilizer formalism.
To do the map from $\ket{\psi(y)}$ to $\ket{\psi(y+\vec{e}_k)}$, we need to flip $y_k$ and the corresponding qubits of $\Lambda y$ in the block $\ket{Wx+(y,\Lambda y)}$.
This can be done by an X-type operator, which we will denote as $\bar{X}'$.
However, we also need to change the coefficient functions $f_j(x)$ correspondingly.
By changing $y$ to $y+\vec{e}_k$, we will change $\vec{\mu}_j$ to $\vec{\mu}_{j'}=\mu_j+(\vec{e}_k, \Lambda \vec{e}_k)$.
Assume that $\bar{X}'D$ achieves the task
\begin{equation}
    \bar{X}'
    D
    \sum_{x\in
          \mathbb{Z}_2^t}
    f_j(x)\mkern2mu
    \ket{x,
         W
         x+
         \vec{\mu}_j}
    =\smashoperator[l]{\sum_{x\in
                             \mathbb{Z}_2^t}}
     f_{j'}(x)\mkern2mu
     \ket{x,
          W
          x+
          \vec{\mu}_{j'}}.
\end{equation}
Then $D$ can be found by noticing that $f_j(x)$ depends on $\vec{\mu}_j$ by the relation~\eqref{eq:coefficient depend on lambda} (note that in relation~\eqref{eq:coefficient depend on lambda}, $\mu_j$ is the $j$-th coordinate of $\vec{\mu}$).
We can compute straightforwardly that
\begin{equation}
    \frac{f_{j'}(x)}
         {f_j(x)}
    =\mathrm{i}^{l(x)}\mkern2mu
     (-1)^{q(x)}
     \prod_h
     (-1)^{\mu_{j
                h}\mkern2mu
           l(x)},
\end{equation}
where $\mu_{jh}$ is the $h$-th coordinate of $\vec{\mu}_j$.
The terms $\mathrm{i}^{l(x)}\mkern2mu(-1)^{q(x)}$ can be obtained by including $S$ and $\ControlledZ$ on the corresponding qubits in $D$.
The terms $(-1)^{\mu_{jh}\mkern2mu l(x)}$ can also be obtained by $Z$ and $\ControlledZ$ by noticing $\mu_{jh}$ is the parity of some qubits in $\ket{x, Wx+\vec{\mu}_j}$.
Thus we have found $\bar{X}=\bar{X}'D$.

\subsection{A Stronger Characterization}
\label{sec:stronger characterization}

Consider an XS-stabilizer state $|\psi\rangle$.
We have shown that $|\psi\rangle$ has the form given in theorem \ref{thm:form_of_stab_states}.
However, not all covariant phases $f\in \CovariantFunctions$ are valid amplitudes.
Here we characterize precisely the subclass of valid amplitudes.

First we note that the group $\PauliSGroupMany$ is closed under conjugation by~$X$, $\sqrt{S}$ and $\ControlledZ$, hence we can make two assumptions about the state:
\begin{enumerate}
    \item We can assume $\vec{\mu}=0$ in theorem \ref{thm:form_of_stab_states}, since we can apply $X$ to the corresponding qubits and update the stabilizers by conjugation.
    \item We can restrict ourselves to studying covariant phases of the form
        \begin{equation}
            \label{eq:nontrivial part of coefficients}
            f(x)
            =\mathrm{i}^{\sum_{j
                               <k}
                         \zeta_{j
                                k}
                         x_j
                         x_k}\mkern2mu
             (-1)^{\sum_{j
                         <k
                         <l}
                   \zeta_{j
                          k
                          l}
                   x_j
                   x_k
                   x_l}
        \end{equation}
        where $\zeta_{jk}$ and $\zeta_{jkl}$ take values in $\mathbb{Z}_2$.
        Indeed, if $f$ is a valid amplitude, then so is
        \begin{equation}
            f(x)\mkern2mu
            \alpha^{l(x)}\mkern2mu
            \mathrm{i}^{l'(x)}\mkern2mu
            (-1)^{q(x)}
        \end{equation}
        for all linear polynomials~$l$, $l'$ and quadratic polynomials~$q$.
        After all, we can always generate these additional phases by applying suitable combinations of the gates $\sqrt{S}$ and $\ControlledZ$ to the state $|\psi\rangle$.
\end{enumerate}
We will treat this class of $f$ as a vector space $V_1$ over $\mathbb{Z}_2$,
with a natural basis given by the functions $i^{x_jx_k}$ and $(-1)^{x_jx_kx_l}$.
Similarly, we consider the set of all functions of the form
\begin{equation}
    (-1)^{\sum_{j
                <k}
          \eta_{j
                k}
          x_j
          x_k}
\end{equation}
which can also be viewed as a vector space $V_2$ over $\mathbb{Z}_2$ with basis functions $(-1)^{x_jx_k}$.
We define a set of linear mappings $\{F_h\mid h=1, \dots, t\}$ from $V_1$ to $V_2$ by the rules
\begin{equation}
    \label{eq:derivative of -1 cubic}
    F_h\mkern2mu
    (-1)^{x_j
          x_k
          x_l}
    =\begin{cases}
         (-1)^{x_k
               x_l}
         & \text{if $h=j$}, \\
         (-1)^{x_j
               x_l}
         & \text{if $h=k$}, \\
         (-1)^{x_j
               x_k}
         & \text{if $h=l$}, \\
         1
         & \text{otherwise},
     \end{cases}
\end{equation}
and
\begin{equation}
    \label{eq:derivative of i quadratic}
    F_h\mkern2mu
    \mathrm{i}^{x_j
                x_k}
    =\begin{cases}
         (-1)^{x_j
               x_k}
         & \text{if $h=j$ or $h=k$}, \\
         1
         & \text{otherwise}.
     \end{cases}
\end{equation}

Recall the matrix $W$ that appears in equation~\eqref{eq:the W matrix}.
Let $\vec{w}_j= (w_{j1}, \dots, w_{jt})$ denote the $j$-th row of $W$, for every $j=1,\dots,n-t$.
We define the quadratic functions $\gamma_j\in V_2$ by
\begin{equation}
    \label{eq:definition of gamma_j}
    \gamma_j(x)
    =(-1)^{\sum_{k
                 <l}
           w_{j
              k}
           w_{j
              l}
           x_k
           x_l}.
\end{equation}
This definition stems from the fact that when we apply the $S$ gate on the single-qubit standard basis state described by $\ket{w_{j1}x_1\oplus\dots\oplus w_{jt}x_t}$, we obtain the phase
\begin{equation}
    \label{eq:discussion_gamma_j}
    \mathrm{i}^{w_{j
                   1}
                x_1\oplus
                \dots\oplus
                w_{j
                   t}
                x_t}
    =\mathrm{i}^{\sum_k
                 w_{j
                    k}
                 x_k}\mkern2mu
     \gamma_j(x),
\end{equation}
where we have used lemma \ref{lem:parity to high order}.
Then we set
\begin{equation}
    \label{eq:definition of Gamma}
    \Gamma
    =\operatorname{span}
     \set{\gamma_j}.
\end{equation}
We will prove the following theorem

\begin{theorem}
    \label{thm:condition for f}
    Consider any state of the form
    \begin{equation}
        \ket{\psi}
        =\frac{1}
              {\sqrt{2^t}}
         \sum_{x\in
               \mathbb{Z}_2^t}
         f(x)\mkern2mu
         \ket{x,
              W
              x}
    \end{equation}
    with $f\in V_1$.
    If $f$ satisfies
    \begin{equation}
        \label{eq:condition for f}
        F_h(f)\in
        \Gamma
        \qquad
        \text{for all $1\leq h\leq t$}
    \end{equation}
    then $|\psi\rangle$ is an XS-stabilizer state.
\end{theorem}
\begin{proof}
    Assuming that $f$ satisfies condition~\eqref{eq:condition for f}, we will show how to construct a set of XS-stabilizers that uniquely stabilize $|\psi\rangle$.
    Consider XS-operators
    \begin{equation}
        g_j
        =X_j
         X(\vec{a}_j)
         S(\vec{b}_j),
        \qquad
        1\leq
        j\leq
        t,
    \end{equation}
    where $X_j$ denotes the Pauli matrix $X$ acting on the $j$-th qubit, where the $X(\vec{a}_j)$ are X-type operators that only act on qubits $t+1$ to $n$ with $\vec{a}_j$ the $j$-th column of the matrix $W$.
    The strings $\vec{b}_j$ are at the moment unspecified.
    Furthermore, define
    \begin{equation}
        g_j
        =Z(\vec{c}_j)
        \qquad
        t+
        1
        \leq
         j
        \leq
         n
    \end{equation}
    where $\{\vec{c}_{t+1}, \dots, \vec{c}_n\}\subseteq\mathbb{Z}_2^n$ form a basis of the orthogonal complement of the subspace
    \begin{equation}
        \OrbitSpace
        =\setbuilder{(x,
                      W
                      x)}
                    {x\in
                     \mathbb{Z}_2^t}
        \subseteq
         \mathbb{Z}_2^n.
    \end{equation}
    For every $y, z\in\mathbb{Z}_2^n$ we have
    \begin{equation}
        Z(z)\mkern2mu
        \ket{y}
        =(-1)^{z^T
               y}\mkern2mu
         \ket{y}.
    \end{equation}
    This implies that $g_j|\psi\rangle = |\psi\rangle$ for every $j=t+1,\dots,n$.
    Next we show that, for a suitable choice of $\vec{b}_j$, the operator $g_j$ stabilizes $|\psi\rangle$ for every $j=1,\dots,t$.
    This last condition is equivalent to $X_jX(\vec{a}_j)|\psi\rangle = S(\vec{b}_j)|\psi\rangle$.
    Note that
    \begin{equation}
        X_j
        X(\vec{a}_j)\mkern2mu
        \ket{\psi}
        =\sum_x
         f(x)\mkern2mu
         \ket{x+
              e_j,
              W
              (x+
               e_j)}
        =\sum_x
         f(x+
           e_j)\mkern2mu
         \ket{x,
              W
              x}
    \end{equation}
    since $\vec{a}_j$ is the $j$-th column of $W$ and thus $\vec{a}_j = We_j$.
    Thus, the condition $g_j|\psi\rangle = |\psi\rangle$ is equivalent to
    \begin{equation}
        f(x+
          e_j)\mkern2mu
        \ket{x,
             W
             x}
        =S(\vec{b}_j)\mkern2mu
         f(x)\mkern2mu
         \ket{x,
              W
              x}.
    \end{equation}
    By using the fact that $f$ has the form~\eqref{eq:condition for f} and by applying the definition of $F_j$, it is easy to check that we have
    \begin{equation}
        \frac{f(x+
                e_j)}
             {f(x)}
        =\mathrm{i}^{l_j(x)}\mkern2mu
         (F_j\circ
          f)(x).
    \end{equation}
    for some linear function $l_j$.
    Summarizing so far, we find that $g_j|\psi\rangle = |\psi\rangle$ if and only if
    \begin{equation}
        \label{eq:S (b_j) on x}
        S(\vec{b}_j)\mkern2mu
        \ket{x,
             W
             x}
        =\mathrm{i}^{l_j(x)}\mkern2mu
         F_j(f)\mkern2mu
         \ket{x,
              W
              x}.
    \end{equation}
    Since, by assumption, we have
    \begin{equation}
        F_j(f)\in
        \Gamma,
    \end{equation}
    we can find a vector $\vec{b}_j'\in\mathbb{Z}_2^{n-t}$, such that
    \begin{equation}
        \prod_{1
               \leq
                k
               \leq
                n-
                t}
        \gamma_k^{b_{j
                     k}'}
        =F_j(f).
    \end{equation}
    This in turn means that~\eqref{eq:S (b_j) on x} is equivalent to
    \begin{equation}
        \label{eq:fincal_condition_bj}
        S(\vec{b}_j)\mkern2mu
        \ket{x,
             W
             x}
        =\mathrm{i}^{l_j(x)}
         \prod_{1
                \leq
                 k
                \leq
                 n-
                 t}
         \gamma_k(x)^{b_{j
                         k}'}
         \ket{x,
              W
              x}.
    \end{equation}
    We now claim that a string $\vec{b}_j$ satisfying this condition exists.
    To see this, first recall~\eqref{eq:discussion_gamma_j} and the surrounding discussion, which implies that
    \begin{equation}
        S(\vec{b}_j')\mkern2mu
        \ket{W
             x}
        =\mathrm{i}^{l_j'(x)}
         \prod_{1
                \leq
                 k
                \leq
                 n-
                 t}
         \gamma_k(x)^{b_{j
                         k}'}
         \ket{W
              x}
    \end{equation}
    for some linear function $l_j'$.
    Second, there exists a string $b_j''\in\{0, 1, 2, 3\}^t$ such that $S(\vec{b}_j'')|x\rangle = i^{l_j(x) - l_j'(x)}|x\rangle$.
    This shows that
    \begin{equation}
        S(\vec{b}_j)
        \coloneqq
         S(\vec{b}_j'')\otimes
         S(\vec{b}_j')
    \end{equation}
    satisfies the condition~\eqref{eq:fincal_condition_bj}.
    We have shown that the operators $g_j$ ($j=1,\dots,n$) stabilize $|\psi\rangle$.
    
    Finally we show that $|\psi\rangle$ is uniquely stabilized by these operators.
    Let $G$ be the group generated by $\{g_1, \dots, g_n\}$.
    Since the $\vec{c}_k$ form a basis of~$\OrbitSpace^\perp$, any string $v\in \mathbb{Z}_2^n$ belongs to~$\OrbitSpace$ if and only if $\vec{z}_j^T v = 0$.
    Thus we have
    \begin{equation}
        \label{eq:condition_support'}
        v\in
        \OrbitSpace
        \quad
        \Leftrightarrow
        \quad
        \text{$g_j\mkern2mu\ket{v}=\ket{v}$ for every $j=t+1,\dots,n$}.
    \end{equation}
    Let $\DiagonalSubgroup$ be the diagonal subgroup of $G$.
    Let $U$ be the set of all $u\in \mathbb{Z}_2^n$ satisfying $D|u\rangle = |u\rangle$ for all $D\in\DiagonalSubgroup$.
    According to the argument above lemma \ref{thm:dimension}, $U$ is the disjoint union of cosets of~$\OrbitSpace$; the number of such cosets is the dimension of the code stabilized by $G$.
    We show that in fact $U=\OrbitSpace$, implying that this dimension is one, so that $|\psi\rangle$ is the unique stabilized state.
    Since each $g_j$ $(j= t+1, \dots, n)$ belongs to~$\DiagonalSubgroup$, every $u\in U$ must satisfy $g_j|u\rangle = |u\rangle$ for all $j=t+1, \dots, n$.
    With \eqref{eq:condition_support'} this shows that $U\subseteq\OrbitSpace$.
    Furthermore, since $D|\psi\rangle= |\psi\rangle$ for every $D\in\DiagonalSubgroup$ and since $|\psi\rangle$ has the form
    \begin{equation}
        \label{eq:psi_support_coset'}
        \ket{\psi}
        =\smashoperator[l]{\sum_{v\in
                                 \OrbitSpace}}
         \alpha_v
         \ket{v}
    \end{equation}
    for some coefficients $\alpha_v$, it follows from that every $v\in\OrbitSpace$ satisfies $D|v\rangle = |v\rangle$ for all $D\in\DiagonalSubgroup'$.
    This shows that $\OrbitSpace\subseteq U$ and thus $\OrbitSpace=U$.
\end{proof}
\begin{remark}
    Since $F_h$ are linear transformations from $V_1$ to $V_2$, and $\Gamma$ is a linear subspace, the condition $F_h(f)\in \Gamma$ can be written as linear equations.
\end{remark}

Note that there will be multiple ways to write down the same XS-stabilizer state.
Let us consider the following example with five qubits:
\begin{equation}
    \ket{\psi}
    =\sum_x
     f(x)\mkern2mu
     \ket{x_1,
          x_2,
          x_3,
          x_1\oplus
          x_2,
          x_2\oplus
          x_3}.
\end{equation}
Equally, we can set $(x_1',x_2',x_3')=(x_1,x_1\oplus x_2, x_2\oplus x_3)$, and the state becomes
\begin{equation}
    \ket{\psi}
    =\sum_{x'}
     f'(x')\mkern2mu
     \ket{x_1',
          x_1'\oplus
          x_2',
          x_1'\oplus
          x_2'\oplus
          x_3',
          x_2',
          x_3'}.
\end{equation}
We can define $\Gamma'$ in the same way as we defined $\Gamma$.
We will show that if $F_{x_j}f(x)\in \Gamma$, then $F_{x_j'}f'(x')\in \Gamma'$.

More formally, we consider the state $\ket{\psi}=\sum_x f(x)\ket{x,Wx}$.
For an invertible matrix $R$ over $\mathbb{Z}_2$ and $x'=R^{-1}x$, we have
\begin{align}
    \ket{\psi}
    & =\sum_x
       f(x)\mkern2mu
       \ket{R
            R^{-1}
            x,
            W
            R
            R^{-1}
            x}
       \nonumber \\
    & =\sum_{x'}
       f(R
         x')\mkern2mu
       \ket{R
            x',
            W
            R
            x'}
       \nonumber \\
    & \equiv
       \sum_{x'}
       f'(x')\mkern2mu
       \ket{W'
            x'}.
\end{align}
In the above equation, we can change the summation from over $x$ to $x'$ because $R$ is invertible.
Then we define $\Gamma'$ from $W'$ in the same way as~\eqref{eq:definition of gamma_j} and~\eqref{eq:definition of Gamma}.
We have the following theorem

\begin{theorem}
    \label{thm:stronger_characterization_converse}
    Let $\ket{\psi}$ be an XS-stabilizer state in the form given in theorem \ref{thm:form_of_stab_states}.
    Then $f$ satisfies the condition~\eqref{eq:condition for f}.
    What is more, for any invertible matrix $R$ over $\mathbb{Z}_2$, the function $f'$ defined by $f'(x')\coloneqq f(Rx')$ also satisfies 
    \begin{equation}
        F_{x_j'}
        f'\in
        \Gamma'.
    \end{equation}
\end{theorem}
\begin{proof}
    First, we note that we have slightly abused the notation here, since $f'$ is in general not a function in $V_1$ (which is defined in section \ref{sec:stronger characterization}).
    However, we can simply ignore the terms $\alpha^{l(x')}\mkern2mu(-1)^{q(x')}$ in $f'$, again by the reasoning in section \ref{sec:stronger characterization}.
    After the transformation $x'=R^{-1}x$, the state can be written as
    \begin{equation}
        \ket{\psi}
        =\sum_{x'}
         f'(x')\mkern2mu
         \ket{R
              x',
              W
              R
              x'}.
    \end{equation}
    Note that if we apply $g_1'=\prod_{j=1}^t g_j^{R_{j1}}$ on the component $f'(x')\mkern2mu\ket{Rx',WRx'}$, we would have
    \begin{equation}
        g_1'\mkern2mu
        f'(x')\mkern2mu
        \ket{R
             x',
             W
             R
             x'}
        =f'(x'+
            \vec{e}_1)\mkern2mu
         \ket{R
              (x'+
               \vec{e}_1),
              W
              R
              (x'+
               \vec{e}_1)},
    \end{equation}
    where $\vec{e}_1$ is the first canonical basis vector.
    By the same reasoning we used in the proof of theorem \ref{thm:condition for f}, we know that
    \begin{equation}
        \frac{f'(x'+
                 \vec{e}_1)}
             {f'(x')}
        =\mathrm{i}^{l(x')}\mkern2mu
         F_{x_1'}[f'](x'),
    \end{equation}
    while at the same time
    \begin{equation}
        g_1'\mkern2mu
        \ket{R
             x',
             W
             R
             x'}
        =\mathrm{i}^{l'(x')}\mkern2mu
         h(x)\mkern2mu
         \ket{R
              (x'+
               \vec{e}_1),
              W
              R
              (x'+
               \vec{e}_1)},
    \end{equation}
    where $h(x)\in\Gamma'$.
    Thus $F_{x_1'}[f']\in\Gamma'$.
    Similarly we can show $F_{x_j'}[f']\in\Gamma'$.
\end{proof}

To illustrate how theorem \ref{thm:condition for f} works, we give two examples here, which demonstrate extreme cases.
First consider the state
\begin{equation}
    \ket{\psi}
    =\sum_x
     f(x)\mkern2mu
     \ket{x_1,
          x_2,
          \dots,
          x_n}.
\end{equation}
By definition, $\Gamma$ is a trivial vector space.
It is then straightforward to check that $f(x)$ has to be of the form $f(x)=i^{l(x)}(-1)^{q(x)}$, and thus $\ket{\psi}$ is a Pauli stabilizer state.
On the other hand, consider the state
\begin{equation}
    \ket{\psi}
    =\sum_x
     f(x)\mkern2mu
     \ket{x_1,
          \dots,
          x_t}
     \bigotimes_{j
                 <k
                 \leq
                  t}
     \ket{x_j\oplus
          x_k}.
\end{equation}
It is easy to check $\Gamma$ is the full vector space $V_2$.
Thus the condition~\eqref{eq:condition for f} becomes trivial, which means $f(x)$ can be an arbitrary function in $\CovariantFunctions$.

\section{Entanglement}
\label{sec:entanglement}

\subsection{Bipartite Entanglement}

In this section we study the bipartite entanglement in XS-stabilizer states.
We consider an $n$-qubit XS-stabilizer state $|\psi\rangle$ (with regular stabilizer group) in the form given in theorem \ref{thm:form_of_stab_states}.
Thus we have
\begin{equation}
    \ket{\psi}
    =\sum_{x\in
           \mathbb{Z}_2^t}
     f(x)\mkern2mu
     \ket{x,
          W
          x}
    =\sum_x
     f(x)\mkern2mu
     \ket{W'
          x}
\end{equation}
where we have denoted
\begin{equation}
    W'
    \coloneqq
     \begin{bmatrix}
         \IdentityMatrix_t \\
         W
     \end{bmatrix}
\end{equation}
which is an $n\times t$ matrix.
The function $f$ belongs to $\CovariantFunctions$ and can be evaluated efficiently owing to theorem \ref{thm:form_of_stab_states}.

Let $(A, B)$ be a bipartition of the qubits.
For convenience we call the two parties Alice and Bob.
We will show that there exists a diagonal operation $D_A\otimes D_B$ mapping the state $|\psi\rangle$ to a Pauli stabilizer state.

There exist permutation matrices $P_A$ and $P_B$ acting in $A$ and $B$, respectively, and an invertible matrix $R$ of dimension $t$, such that
\begin{equation}
    \begin{bmatrix}
        P_A & \\
            & P_B
    \end{bmatrix}
    W'
    R^{-1}
    =\left[\begin{array}{@{}cc@{}}
               \IdentityMatrix_k &           0           \\
                       M         &           0           \\
               \hline
                       0         & \IdentityMatrix_{t-k} \\
                      C_1        &          C_2
           \end{array}\right]
\end{equation}
for some $k$ and some $M$, $C_1$ and $C_2$ (and where, in the r.h.s, the upper block refers to the qubits in $A$ and the lower block refers to the qubits in $B$).
We carry out a change of variables $x\mapsto x'=Rx$.
Furthermore we denote $x'=(u, v)$ where $u$ denotes the first $k$ bits of $x'$.
Also, we set $f'(x')\coloneqq f(x)$.
We can then write the state in the form
\begin{equation}
    \ket{\psi}
    =\sum_{x'
           =(u,
             v)}
     f'(x')\mkern2mu
     \ket{u,
          M
          u}_A\otimes
     \ket{v,
          C_1
          u+
          C_2
          v}_B,
\end{equation}
where $u$ ranges over $\mathbb{Z}_2^k$ and $v$ ranges over $\mathbb{Z}_2^{t-k}$.
Let $r$ denote the rank of $C_1$.
Notice that there exists a full-rank matrix $D_1$ (its rank being $r$) and an invertible matrix $T$ such that $C_1= [D_1\ |\ 0] T$ where $0$ is the zero matrix of appropriate dimensions.
This means there is a further change of variables from $u$ to $w=Tu$ such that only the first $r$ bits of $w$ appear on Bob's side.
In a more explicit way, we can rewrite the computational basis states appearing on Bob's side as
\begin{equation}
    \label{eq:standard form on bob side}
    \ket{v,
         D_1(w_1,
             \dots,
             w_r)+
         C_2
         v}.
\end{equation}
And it is easy to see that the form of the state on Alice's side still does not involve $v$, since it can be written as:
\begin{equation}
    \ket{T
         w,
         M
         T
         w}.
\end{equation}
We write $x''=(w, v)$ and $f''(x'')\coloneqq f'(x')$.
Note that the function $f''$ is related to the function $f$ by a linear change of variables.
Hence, owing to theorem \ref{thm:stronger_characterization_converse}, $f''$ must also satisfy condition~\eqref{eq:condition for f}.
We will use this condition to gain insight in the form of $f''$.
We will only consider the non-trivial part $h(x'')$ of $f''(x'')$, as defined in~\eqref{eq:nontrivial part of coefficients}.
The reason we only focus on $h$ is because we can obtain all linear terms $\alpha^{l(x'')}$ by acting locally in each party, and we can simply ignore $(-1)^{q(x'')}$ because these terms are all allowed for a Pauli stabilizer state (recall that our goal is to map $|\psi\rangle$ to a Pauli stabilizer state by means of an operation $D_A\otimes D_B$).
We claim $h$ only depends on $w_1$, …, $w_r$ and $v$.
We prove this fact by contradiction.
For example, assume that $h$ contains a term $(-1)^{w_{r+1} v_{a}v_{b}}$.
Then by equations~\eqref{eq:derivative of -1 cubic} and~\eqref{eq:derivative of i quadratic}, we know
\begin{equation}
    (F_a\circ
     h)(w,
        v)
    =(-1)^{w_{r+
               1}
           v_b+
           q(w,
             v)},
\end{equation}
where $q(w,v)$ is a quadratic function that does not contain the term $w_{r+1}v_b$.
By theorem \ref{thm:condition for f}, we know that
\begin{equation}
    F_a(h)\in
    \Gamma.
\end{equation}
However, by observing~\eqref{eq:standard form on bob side}, we notice that none of the functions $\gamma_j$ (which are defined in~\eqref{eq:definition of gamma_j}) can contain $(-1)^{w_{r+1}v_b}$.
This implies that
\begin{equation}
    F_a(f'')\not\in
    \Gamma,
\end{equation}
which leads to the contradiction.
Similarly we can show that there are no $i^{w_jv_b}$ terms in $h$ for $j\geq r+1$.

We have shown that the function $h$ only depends on $w_1$, …, $w_r$ and $v$.
This implies that, by acting with a suitable diagonal unitary operation within Bob's side, we can remove the corresponding phase in the state $|\psi\rangle$; the resulting state is a Pauli stabilizer state.
To see how these phases can be removed locally, we argue as follows.
The standard basis kets on Bob's side have the form \eqref{eq:standard form on bob side} with $D_1$ full rank.
Thus there exists a (full rank) matrix $E$ such that $ED_1w = w$ for every $w= (w_1, \dots, w_r)$.
Let $U_1$ be the unitary operation which implements $E$
\begin{equation}
    U_1\colon
    \ket{v,
         D_1(w_1,
             \dots,
             w_r)+
         C_2
         v}
    \mapsto
     \ket{v,
          (w_1,
           \dots,
           w_r)+
          E
          C_2
          v}
\end{equation}
After applying $U_1$, we apply the operation $U_2$ defined by
\begin{equation}
    U_2\colon
    \ket{v,
         (w_1,
          \dots,
          w_r)+
         E
         C_2
         v}
    \mapsto
     \ket{v,
          (w_1,
           \dots,
           w_r)}
\end{equation}
Note that both $U_1$ and $U_2$ can be realized as circuits of $\ControlledNOT$ gates.
Then we apply the diagonal operation $D$ defined by
\begin{equation}
    D\colon
    \ket{v,
         (w_1,
          \dots,
          w_r)}
    \mapsto
     h(w_1,
       \dots,
       w_r,
       v)^{-1}
     \ket{v,
          (w_1,
           \dots,
           w_r)}.
\end{equation}
Since $h$ can be evaluated efficiently (this follows from the fact that $f$ can be evaluated efficiently), the operation $D$ can be implemented efficiently.
Finally, we apply $U_2^\dagger$ followed by $U_1^\dagger$, yielding
\begin{equation}
    h(w_1,
      \dots,
      w_r,
      v)\mkern2mu
    \ket{D_1(w_1,
             \dots,
             w_r)+
         C_2
         v}.
\end{equation}
This total procedure $D_B\coloneqq U_1^\dagger U_2^\dagger D U_2U_1$ thus allows us to multiply each ket \eqref{eq:standard form on bob side} with the inverse of $h(w_1, \dots, w_r, v)$, thereby \enquote{canceling out} the function $h$.
Note that $D_B$ is a diagonal operation, since $D$ is diagonal and since $U_2U_1$ is a permutation of the standard basis.

We have shown:

\begin{theorem}
    \label{thm:bipartite_reduction}
    Let $|\psi\rangle$ be an XS-stabilizer state and $(A, B)$ a bipartition of its qubits.
    Then there exists a Pauli stabilizer state~$\ket{\phi_{A,B}}$ and diagonal operators~$D_A$ and~$D_B$ such that
    \begin{equation}
        \ket{\psi}
        =D_A\otimes
         D_B\mkern2mu
         \ket{\phi_{A,
                    B}}.
    \end{equation}
    A description of~$\ket{\phi_{A,B}}$ can be computed efficiently.
\end{theorem}

\begin{corollary}
    The von Neumann entanglement entropy of $|\psi\rangle$ w.r.t. $(A, B)$ can be computed efficiently.
\end{corollary}

\begin{corollary}
    Let $\rho_A$ be the reduced density operator of $|\psi\rangle$ for the qubits in $A$.
    Then $\rho_A$ is proportional to a projector i.e. all nonzero eigenvalues of $\rho_A$ coincide.
\end{corollary}
The first corollary holds since the entanglement entropy of Pauli stabilizer states can be computed efficiently~\cite{Fattal2004entanglement}.
The second corollary holds since reduced density operators of Pauli stabilizer states are proportional to projectors~\cite{Hein2006entanglement}.

The state~$\ket{\phi_{A,B}}$ generally depends on the bipartition~$(A,B)$.
It would be interesting to understand whether Theorem~\ref{thm:bipartite_reduction} can be made independent of that:

\begin{problem}
    For every XS-stabilizer state~$\ket{\psi}$, does there exist a single Pauli stabilizer state~$\ket{\phi}$ such that for \emph{every} bipartition~$(A,B)$ we have
    \begin{equation}
        \ket{\psi}
        =U_A\otimes
         U_B\mkern2mu
         \ket{\phi}
    \end{equation}
    for some local unitaries~$U_A$ and~$U_B$?
\end{problem}

\subsection{LU-Inequivalence of XS- and Pauli Stabilizer States}
\label{sec:lu non equivalence}

Here we show that the XS-stabilizer state~\eqref{eq:6qubit_example} is not LU-equivalent to any Pauli stabilizer state, i.e. there does not exist any Pauli stabilizer state~$\ket{\phi}$ satisfying $\ket{\psi}=U\mkern2mu\ket{\phi}$ for any $U\coloneqq U_1\otimes\cdots\otimes U_6\in\mathrm{U}(2)^{\otimes 6}$.
This result demonstrates that there exist XS-stabilizer states whose multipartite entanglement (w.r.t LU-equivalence) is genuinely different from that of any Pauli stabilizer state.

In order to prove the claim, we consider the classification of 6-qubit stabilizer (graph) states under LU-equivalence as given in Fig. 4 and Table II of Ref. \cite{Hein2004multiparty}.
In the latter figure, 11 distinct LU-equivalence classes are shown to exist for fully entangled 6-qubit stabilizer states; the classes are labeled from 9 to 19.
A representative of each class is given in Fig 4.
Furthermore each class is uniquely characterized by its list of Schmidt ranks, i.e. the Schmidt ranks for all possible bipartitions of the system.
In table II it is shown that, for any of the classes 9--17, there is at least one bipartition of the form (two qubits -- rest) for which the state is not maximally entangled.
This shows that the XS-stabilizer state $|\psi\rangle$ cannot be LU-equivalent to any of the states in the classes 9--17, since $|\psi\rangle$ is maximally entangled for all such bipartitions.
Furthermore, for class 19, the entanglement is maximal for all bipartitions of the form (3 qubits -- rest); since this is not the case for $|\psi\rangle$, the latter cannot be LU-equivalent to any state in class 19.
This leaves class 18.
Consider the state
\begin{equation}
    \ket{\phi}
    =\smashoperator[l]{\sum_{x_j
                             =0}^1}
     \ket{x_1,
          x_2,
          x_3,
          x_1\oplus
          x_2,
          x_2\oplus
          x_3,
          x_3\oplus
          x_1}
\end{equation}
which is a Pauli stabilizer state.
By direct computation of all Schmidt ranks, one verifies that this state belongs to class 18.
We prove by contradiction that $|\psi\rangle$ is not LU-equivalent to~$\ket{\phi}$.
First, it is straightforward to show that for both~$\ket{\phi}$ and $|\psi\rangle$, the 3-qubit reduced density matrix of the 1st, 2nd and 4th qubits is
\begin{align}
    \rho_{1
          2
          4}
    & =\frac{1}
            {4}
       (\ketbra{0
                0
                0}
               {0
                0
                0}+
        \ketbra{0
                1
                1}
               {0
                1
                1}+
        \ketbra{1
                0
                1}
               {1
                0
                1}+
        \ketbra{1
                1
                0}
               {1
                1
                0})
       \nonumber \\
    & =\frac{1}
            {8}
       (\Identity+
        Z_1\otimes
        Z_2\otimes
        Z_4).
\end{align}
If there is a local unitary transformation~$U$ from $\ket{\phi}$ to $\ket{\psi}$, then $U_1\otimes U_2\otimes U_4$ must leave $\rho_{124}$ unchanged.
This implies that $U_1\otimes U_2\otimes U_4$ must leave $Z_1\otimes Z_2\otimes Z_4$ unchanged, so that $U_1Z_1U_1^\dagger \propto Z_1$ and similarly for $U_2$ and $U_4$.
This implies that $U_1$, $U_2$ and $U_4$ must have the form~$D_j$ or~$D_jX$ for some diagonal matrix~$D_j$.
Analogously, we can show that the same holds for all other~$j$.
Thus $U=D\mkern2mu X(\vec{a}\mkern1mu)$ for some diagonal operator $D\coloneqq D_1\otimes\dots\otimes D_6$ and some~$\vec{a}\in\mathbb{Z}_2^6$.
Note that $|\psi\rangle$ and~$\ket{\phi}$ have the form
\begin{equation}
    \ket{\phi}
    =\smashoperator[l]{\sum_{v\in
                             \OrbitSpace}}
     \ket{v}
    \qquad
    \text{and}
    \qquad
    \ket{\psi}
    =\smashoperator[l]{\sum_{v\in
                             \OrbitSpace}}
     \alpha_v\mkern2mu
     \ket{v},
\end{equation}
for some linear subspace $\OrbitSpace\subseteq\mathbb{Z}_2^6$ and real coefficients $\alpha_v$.
Then
\begin{equation}
    D\mkern2mu
    X(\vec{a}\mkern1mu)\mkern2mu
    \ket{\phi}
    =\smashoperator[l]{\sum_{v\in
                             \OrbitSpace}}
     \beta_v\mkern2mu
     \ket{v+
          \vec{a}\mkern1mu}
\end{equation}
for some coefficients $\beta_v$.
Thus $U\mkern2mu\ket{\phi}=\ket{\psi}$ implies that $\OrbitSpace=\OrbitSpace+ \vec{a}$.
This shows that $\vec{a}\in\OrbitSpace$.
But then $X(\vec{a}\mkern1mu)\mkern2mu\ket{\phi}=\ket{\phi}$.
The identity $U\mkern2mu\ket{\phi}=\ket{\psi}$ thus implies that $D\mkern2mu\ket{\phi}=\ket{\psi}$.
It is straightforward to verify that this cannot be true.
We have thus shown that $|\psi\rangle$ and~$\ket{\phi}$ are not LU-equivalent.
In conclusion, $|\psi\rangle$ does not belong to any LU-equivalence class of Pauli stabilizer states.

\section{Efficient Algorithms}
\label{sec:algorithms}

In this section we will give a list of problems that can be solved with efficient classical algorithms for regular XS-stabilizer states (codes).
We consider an arbitrary $n$-qubit regular XS-stabilizer stabilizer code $\RegularXSCode$ specified in terms of a generating set of $m$ stabilizers in the standard form given in Corollary~\ref{cor:normal_form_regular}.
Then the following holds:
\begin{enumerate}
    \item The {\bf degeneracy} $d$ of the code can be computed in $\poly(n,m)$ time (recall section \ref{sec:constructing_basis}).
    \item An efficient algorithm exists to {\bf determine $d$ basis states} $|\psi_1\rangle$, …, $|\psi_d\rangle$, each of which is an XS-stabilizer state with regular stabilizer group and each state having the form
        \begin{equation}
            \ket{\psi_i}
            =\smashoperator[l]{\sum_{x\in
                                     \mathbb{Z}_2^t}}
             f_i(x)\mkern2mu
             \ket{x,
                  W
                  x+
                  \vec{\mu}_i}
             \qquad
             \text{with $f_i\in\CovariantFunctions$}.
        \end{equation}
        The matrix $W$ (which is the same for all $\ket{\psi_i}$) can be computed in $\poly(n,m)$ time.
        The list $\{\vec{\mu}_1$, …, $\vec{\mu}_d\}$ can be computed in $\poly(n,m,d)$ time.
        Given a specific $\vec{\mu}_i$, a complete generating set of stabilizer operators having $|\psi_i\rangle$ as unique stabilized state can be computed in $\poly(n,m)$ time.
        Furthermore, given $\vec{\mu}_i$, the function $x\mapsto f_i(x)$ can be computed in $\poly(n,m)$ time as well.
        See section \ref{sec:constructing_basis}.
    \item The {\bf logical operators} of $\RegularXSCode$ can be computed in $\poly(n,m,d)$ time.
        See section \ref{sec:logical_operators}.
    \item The {\bf commuting Hamiltonian} described in section \ref{sec:hamiltonians} can be computed in $\poly(n,m)$ time.
\end{enumerate}

On input of~$\vec{\mu}_i$ the following holds in addition:
\begin{enumerate}
    \setcounter{enumi}{4}
    \item The von Neumann {\bf entanglement entropy} of any $|\psi_i\rangle$ with regular stabilizer group can be computed, for any bipartition, in $\poly(n,m)$ time.
        This claim holds since we have shown in section \ref{sec:entanglement} how to efficiently compute the description of a Pauli stabilizer state with the same entanglement as $|\psi_i\rangle$; furthermore an efficient algorithm to compute the von Neumann entanglement entropy of Pauli stabilizer states is known \cite{Fattal2004entanglement}.
    \item A $\poly(n)$ size {\bf quantum circuit} to generate any $|\psi_i\rangle$ can be computed for in $\poly(n,m)$ time.
        This circuit can always be chosen to be a Clifford circuit followed by a circuit composed of the diagonal gates $\ControlledControlledZ$ (controlled-$\ControlledZ$), $\ControlledS\coloneqq\ControlledSLong$ (controlled-$S$) and $T$.
        To see this, we recall theorem \ref{thm:form_of_stab_states}.
        This implies that the state $|\psi_i\rangle$ can be prepared as follows:
        \begin{itemize}
            \item Using a Clifford circuit $\CliffordCircuit_1$, prepare the state $\sum |x, Wx + \vec{\mu}_i\rangle$.
                In fact, this can be done using a circuit composed of Hadamard, $X$ and $\ControlledNOT$ gates.
            \item Since the function $f_i$ belongs to the class $\CovariantFunctions$, it has the form
                \begin{equation}
                    \alpha^{l(x)}\mkern2mu
                    \mathrm{i}^{q(x)}\mkern2mu
                    (-1)^{c(x)}.
                \end{equation}
                Note that the we have the following gate actions on the standard basis:
                \begin{align*}
                    T
                    & \colon
                       \ket{x}\mapsto
                       \alpha^x\mkern2mu
                       \ket{x}, \\
                    S
                    & \colon
                       \ket{x}\mapsto
                       \mathrm{i}^x\mkern2mu
                       \ket{x}, \\
                    \ControlledS
                    & \colon
                       \ket{x,
                            y}\mapsto
                       \mathrm{i}^{x
                                   y}\mkern2mu
                       \ket{x}, \\
                    \ControlledZ
                    & \colon
                       \ket{x,
                            y}\mapsto
                       (-1)^{x
                             y}\mkern2mu
                       \ket{x,
                            y}, \\
                    \ControlledControlledZ
                    & \colon
                       \ket{x,
                            y,
                            z}\mapsto
                       (-1)^{x
                             y
                             z}\mkern2mu
                       \ket{x,
                            y,
                            z}.
                \end{align*}
                Therefore, the phase $f_i(x)$ can be generated by first applying a suitable circuit $\CliffordCircuit_2$ of Clifford gates $\ControlledZ$ and $S$ to generate the quadratic part of $c(x)$ and the linear part of $q(x)$, and by subsequently applying a suitable circuit $\TCSCCZCircuit$ composed of the (non-Clifford) gates $T$, $\ControlledS$ and $\ControlledControlledZ$ to generate $l(x)$, the quadratic part of $q(x)$ and the cubic part of $c(x)$, respectively.
                Since the function $f_i$ can be computed efficiently, the descriptions of $\CliffordCircuit_2$ and $\TCSCCZCircuit$ can be computed efficiently.
                The overall circuit is $\TCSCCZCircuit\CliffordCircuit_2\CliffordCircuit_1$.
        \end{itemize}
    \item Given any $|\psi_i\rangle$ and Pauli operator $P$, we can {\bf compute the expectation value} $\langle \psi_i|P|\psi_i\rangle$ in $\poly(n,m)$ time.
        This implies in particular that the expectation of any local observable (i.e. an observable acting on a subset of qubits of constant size) can be computed efficiently as well, since every such observable can be written as a sum of $\poly(n)$ Pauli observables.
        To see that $\langle \psi_i|P|\psi_i\rangle$ can be computed efficiently, recall from point 6 above that $|\psi_i\rangle$ can be decomposed as $|\psi_i\rangle = \TCSCCZCircuit|\psi_i'\rangle$ where $\TCSCCZCircuit$ is a circuit composed of $T$, $\ControlledS$ and $\ControlledControlledZ$, and where $|\psi_i'\rangle= \CliffordCircuit_2\CliffordCircuit_1|0\rangle$ is a Pauli stabilizer state.
        Then
        \begin{equation}
            \expval{\psi_i}
                   {P}
                   {\psi_i}
            =\expval{\psi_i'}
                    {\TCSCCZCircuit^\dagger
                     P
                     \TCSCCZCircuit}
                    {\psi_i'}.
        \end{equation}
        Its is easily verified that $\TCSCCZCircuit^\dagger P\TCSCCZCircuit\eqqcolon\CliffordCircuit''$ is a Clifford operation, for every circuit $\TCSCCZCircuit$ composed of $T$, $\ControlledS$ and $\ControlledControlledZ$ (for example $TXT \propto S$).
        Thus we have
        \begin{equation}
            \expval{\psi_i}
                   {P}
                   {\psi_i}
            =\expval{\psi_i'}
                    {\CliffordCircuit''}
                    {\psi_i'}.
        \end{equation}
        Recall that $|\psi_i'\rangle= \CliffordCircuit_2\CliffordCircuit_1|\psi_i\rangle$, we know that
        \begin{equation}
            \expval{\psi_i}
                   {P}
                   {\psi_i}
            =\expval{0}
                    {\CliffordCircuit'''}
                    {0},
        \end{equation}
        where $\CliffordCircuit'''=\CliffordCircuit_1^{\dagger} \CliffordCircuit_2^{\dagger}\CliffordCircuit \CliffordCircuit_2\CliffordCircuit_1$.
        Note that $\expval{0}{\CliffordCircuit'''}{0}$ is simply the coefficient of the basis $\ket{0}$ in the Pauli stabilizer state $\CliffordCircuit'''\ket{0}$, which can be computed efficiently according to~\cite{Vandennest2010classical}.
\end{enumerate}

\section{Non-Regular XS-Stabilizer Groups}
\label{sec:non regular xs stabilizer}

Though we have tried to avoid non-regular XS-stabilizer groups due to the computational hardness, there are situations where they appear naturally.
For example, let us look at~\eqref{eq:non abelian anyon stabilizer1} through~\eqref{eq:non abelian anyon stabilizer3} in the appendix.
They describe a code space $\RegularXSCode$ that is equivalent to the ground space of the twisted quantum double model $\mathrm{D}^\omega(\mathbb{Z}_2 \times\mathbb{Z}_2 \times \mathbb{Z}_2)$ by a local unitary circuit (as defined in \cite{Chen2010local}).
These stabilizer operators have an interesting property: if they are on an infinite lattice or a lattice with open boundary, then they generate a regular XS-stabilizer group.
On the other hand, for example, if they are on a torus, the group they generated will not be regular.
This is related to the fact that this model has a ground state degeneracy of 22 when it is on a torus, which cannot be the degeneracy of a regular XS-stabilizer code.
It is also known that this twisted quantum double model support non-Abelian anyons, which has shown to be impossible for Pauli stabilizer codes on 2D.

Given the existence of interesting non-regular XS-stabilizer groups, we want to make a few comments about which results in this paper still hold for non-regular groups.
First, we have the following theorem:
\begin{theorem}
    \label{thm:non-regular state has regular group}
    Every XS-stabilizer state has a \emph{regular} XS-stabilizer group which uniquely stabilizes it.
\end{theorem}
\begin{proof}
    Let $G$ be the initial (generally non-regular) stabilizer group of $|\psi\rangle$.
    We show that $G$ can be replaced with a regular stabilizer group.
    Let $\{D_1, \dots, D_r\}$ be generators of the diagonal subgroup of $G$.
    Extend this set to a generating set of $G$, say $\Generators=\{D_1 ,\dots, D_r, g_1, \dots, g_k\}$.
    Here each $g_j$ has the form $g_j=i^{s_j} X(\vec{a}_j)S(\vec{b}_j)$.
    We can assume $g_j$ is non-diagonal, and $\vec{a}_j$ are linearly independent of each other.
    Since if this is not the case, we can use the procedure in Theorem~\ref{thm:normal_form} to transform $g_j$ to satisfy this condition.
    Notice that $g_j$ are all monomial unitary matrix.
    It follows that the permutation group $\PermutationGroup$ associated with $G$ is generated by the operators $X(\vec{a}_j)$.
    Let $\OrbitSpace\subseteq \mathbb{Z}_2^n$ (where $n$ denotes the number of qubits) be the linear span of the $\vec{a}_j$.
    Then $\PermutationGroup= \{X(v)\mid v\in\OrbitSpace\}$.
    Furthermore, the orbit of a computational basis state $|x\rangle$ is the coset of~$\OrbitSpace$ containing $x$ i.e. $\Orbit_x= x+\OrbitSpace$.
    
    Consider the set $\InvariantSetDiagonal$ of those $n$-bit strings $z$ satisfying $D|z\rangle = |z\rangle$ for all $D\in\DiagonalSubgroup$.
    Furthermore, recall from the proof of theorem \ref{thm:inside NP} that $G_x =\DiagonalSubgroup$ for every $x$.
    Applying theorem \ref{thm:monomial}(b) and using that the dimension of the space stabilized by $G$ is 1 (since $|\psi\rangle$ is an XS-stabilizer state) we conclude that $\InvariantSetDiagonal=x+\OrbitSpace$ for some $x$, and that $|\psi\rangle$ must have the form
    \begin{equation}
        \ket{\psi}
        =\smashoperator[l]{\sum_{v\in
                                 \OrbitSpace}}
         f(v)\mkern2mu
         \ket{v+
              x}.
    \end{equation}
    Now define Z-type operators $h_k$ of the form $h_k(-1)^{s_k} Z(\vec{b}_k)$ $(k=1, \dots, q)$, where $s_k$ and $\vec{b}_k$ are chosen such that $\InvariantSetDiagonal$ coincides with the set of all $z$ satisfying $\vec{b}_k^Tz = s_k$ for all $k$.
    This means $h_k|z\rangle = |z\rangle$ for all $z\in x+\OrbitSpace$ and in turn $h_k|\psi\rangle = |\psi\rangle$.
    It follows that $|\psi\rangle$ is stabilized by $\Generators':=\{h_1 ,\dots, h_q, g_1, \dots, g_k\}$.
    Finally, by lemma~\ref{lemma:sufficient condition for regular group}, we know the group $G'$ is regular, and by the argument in lemma~\ref{thm:compute_form_psi}, we know $\ket{\psi}$ is uniquely stabilized by $G'$.
\end{proof}

Now consider the procedure in lemma~\ref{thm:compute_form_psi}.
It is easy to see even if the group $G$ is non-regular, as long as we have a $\vec{\lambda}_j\in\InvariantSetDiagonal$, we can still find the $\ket{\psi_j}$ corresponds to $\vec{\lambda}_j$.
By theorem \ref{thm:monomial}, there is a set of $\{\vec{\lambda}_j\}$ such that the corresponding $\ket{\psi_j}$ form a basis for the space stabilized by $G$.
Again by the procedure in lemma~\ref{thm:compute_form_psi}, we know $G$ can be expanded to uniquely stabilize each $\ket{\psi_j}$.
Thus by theorem~\ref{thm:non-regular state has regular group} we know $\ket{\psi_j}$ is a regular XS-stabilizer state.
This means although it is (computationally) hard to find $\vec{\lambda}_j$, the basis $\ket{\psi_j}$ for the code space still satisfies all the properties we proved, including the form of the phases $f(x)$ and the bipartite entanglement.
The construction of commuting Hamiltonian also does not require the stabilizer group to be regular.

On the other hand, for non-regular stabilizer groups, there is no general formula for the degeneracy.
We also cannot find logical operators that have a similar form as the ones in section~\ref{sec:logical_operators}, since the degeneracy of $\RegularXSCode$ is not necessarily $2^k$.

\section{Open Questions}

In this section we will summarize a few interesting questions about the XS-stabilizer formalism, some of which have already been mentioned in the text.
\begin{description}
    \item[The group structure]
        While the tractability of XS-stabilizer states $\ket{\psi_j}$ is closely related to the fact that each XS-stabilizer group $G$ is a rather particular finite group, the properties of $\ket{\psi_j}$ are not.
        It would be interesting to establish some direct link between the group $G$ and the states $\ket{\psi_j}$ (e.g. a relation between the reduced density matrix $\rho$ and $G$).
    \item[Properties of entanglement]
        As we mentioned in section~\ref{sec:entanglement}, it is not known whether for any XS-stabilizer state $\ket{\psi}$ there exists a single Pauli stabilizer state $\ket{\varphi}$ that has the same von Neumann entropy across \emph{all} bipartitions.
        It would also be interesting to know to what extent the inequalities described in~\cite{Linden2013quantum} hold for XS-stabilizer states.
    \item[Logical operators and transversal gates]
        We have shown how to construct $\bar{Z}_j$ and $\bar{X}_j$ operators in section~\ref{sec:logical_operators}.
        The $\bar{Z}_j$ are transversal gates by definition.
        While we showed that the $\bar{X}_j$ operators include~$X$, $S$, and~$\ControlledZ$ in general, it is possible that for many codes the~$\bar{X}_j$ only contain $X$ and $S$.
        In particular, $S$ and $CZ$ are interchangeable in some cases.
        For example, consider the state
        \begin{equation}
            \ket{\psi}
            =\sum_{x_1,
                   x_2}
             \ket{x_1,
                  x_2,
                  x_1\oplus
                  x_2}.
        \end{equation}
        It is easy to check that
        \begin{equation}
            \ControlledZ_{1
                          2}\mkern2mu
            \ket{\psi}
            =S_1^3
             S_2^3
             S_3\mkern2mu
             \ket{\psi}.
        \end{equation}
        Thus it would be interesting to know when a certain XS-stabilizer code has transversal $\bar{X}_j$ operators, and possibly some other transversal gates.
    \item[Quantum phases]
        Understanding topological phases is an extremely important but also very hard task.
        Compared to general local Hamiltonians, the Hamiltonians generated by local Pauli stabilizer codes are much easier to analyze.
        Thus the Pauli stabilizer formalism has proved a gateway both to studying the behaviour of topological phases and to constructing new models.
        It is then natural to ask whether we can classify all topological phases described by XS-stabilizer codes or whether we can construct new models in 2D and 3D.
    \item[Non-regular XS-stabilizer]
        As we have shown in this paper, it is in general computationally hard to study the states stabilized by non-regular XS-stabilizer groups.
        Restricting to regular groups is sufficient to circumvent this problem, but not necessary.
        It is thus desirable to find the necessary conditions under which the XS-stabilizer problem will become efficient.
        For example, it is not clear whether the XS-stabilizer problem is still hard if the number of S-type operators in the generators of the diagonal subgroup is constant.
\end{description}

\section{Acknowledgements}

This research was supported in part by Perimeter Institute for Theoretical Physics.
Research at Perimeter Institute is supported by the Government of Canada through Industry Canada and by the Province of Ontario through the Ministry of Research and Innovation.

\appendix

\section{Twisted Quantum Double Models}
\label{sec:twisted_quantum_double}

\begin{figure}
    \centering
    \includegraphics[scale=0.25]{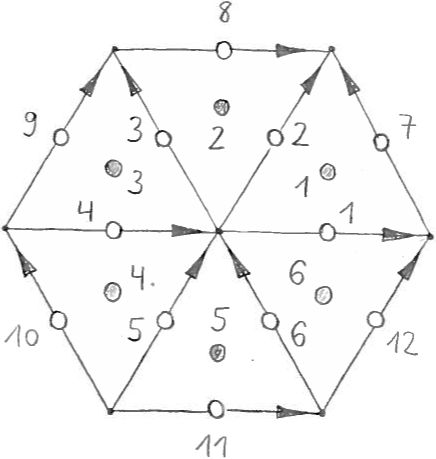}
    \caption{A branching structure on the triangular lattice.
        White circles denote qudits~$\ket{x_i}$, grey circles denote possible ancilla qudits~$\ket{y_p}$.}
    \label{fig:branching}
\end{figure}

We study the twisted quantum double models~$\mathrm{D}^\omega(\mathbb{Z}_2^n)$ with the groups~$\mathbb{Z}_2^n$ and twists $\omega\in H^3\bigl(\mathbb{Z}_2^n,U(1)\bigr)$ on a triangular lattice.
Although every such group is Abelian, for certain~$n$ and~$\omega$ the twisted quantum double model~$\mathrm{D}^\omega(\mathbb{Z}_2^n)$ will harbour \emph{non}-Abelian anyons as excitations.

Without loss of generality we choose the branching structure shown in Figure~\ref{fig:branching} for the triangular lattice.
Each lattice edge~$i$ carries a Hilbert space with basis $\{\ket{x_i}\mid x_i\in\mathbb{Z}_2^n\}$.
By abuse of notation $\ket{x_i}$ is either the state of an actual qubit if $n=1$ or the state of a qudit if $n>1$.
In the latter case we write elements $t=(t_1,\dots,t_n)\in\mathbb{Z}_2^n$ as binary strings over the alphabet~$\{0,1\}$ and accordingly expand the qudit state~$\ket{x_i}=\ket{x_{i,1},\dots,x_{i,n}}$ in terms of qubit states~$\ket{x_{i,\sigma}}$ where $i$ denotes the position on the lattice and $\sigma$ the \enquote{layer}.
Furthermore we write group multiplication in~$\mathbb{Z}_2^n$ additively.

The Hamiltonian is given as a sum of commuting projectors:
\begin{equation}
    H
    =-\sum_s
     A^\omega(s)
     -\sum_p
     B(p).
\end{equation}
Each operator~$B(p)$ is associated with a triangle~$p$ of the lattice and reads
\begin{equation}
    B(p)
    =\delta(x_i+
            x_j+
            x_k)\mkern2mu
     \ketbra{x_i,
             x_j,
             x_k}
            {x_i,
             x_j,
             x_k}
\end{equation}
where $i$, $j$ and~$k$ denote the edges of~$p$.
It enforces a flat connection on the triangle~$p$ in the ground state subspace.
The operator~$A^\omega(s)$ associated with a vertex~$s$ is defined by
\begin{equation}
    A^\omega(s)
    =\frac{1}
          {2^n}
     \sum_{t\in
           \mathbb{Z}_2^n}
     A_t^\omega(s).
\end{equation}
If $s$ is the central vertex of Figure~\ref{fig:branching} the individual terms are given by
\begin{equation}
    A_t^\omega(s)
    =\smashoperator[l]{\sum_{x_i\in
                             \mathbb{Z}_2^n}}
     f_t^\omega(x)\mkern2mu
     \ketbra{x_1+
             t,
             \dots,
             x_6+
             t}
            {x_1,
             \dots,
             x_6}\otimes
     \ketbra{x_7,
             \dots,
             x_{12}}
            {x_7,
             \dots,
             x_{12}}
\end{equation}
with the phases%
\footnote{Note that these phases do not explicitly depend on the values~$x_2$, $x_5$, $x_9$ and~$x_{12}$.
    This may change if one fixes a different branching structure on the triangular lattice.}
\begin{equation}
    f_t^\omega(x)
    =\frac{\omega(t,
                  x_4,
                  x_{10})\mkern2mu
           \omega(x_3+
                  t,
                  t,
                  x_4)\mkern2mu
           \omega(x_8,
                  x_3+
                  t,
                  t)}
          {\omega(t,
                  x_6,
                  x_{11})\mkern2mu
           \omega(x_1+
                  t,
                  t,
                  x_6)\mkern2mu
           \omega(x_7,
                  x_1+
                  t,
                  t)}
    =\pm1.
\end{equation}
Note that each~$\omega$ couples two distinct qudit variables~$x_i$ and~$x_j$ which always belong to some triangle.
Also, these phases enjoy the property
\begin{equation}
    f_{t
       t'}^\omega(x)
    =f_t^\omega(x_1+
                t',
                \dots,
                x_6+
                t',
                x_7,
                \dots,
                x_{12})\mkern2mu
     f_{t'}^\omega(x)
\end{equation}
which implies $A_t^\omega(s)\mkern2mu A_{t'}^\omega(s)=A_{tt'}^\omega(s)$.
The phases arising from a product of 3-cocycles~$\omega$ and~$\omega'$ factorize as
\begin{equation}
    \label{eq:phases_product_cocycle}
    f_t^{\omega
         \omega'}(x)
    =f_t^{\omega}(x)\mkern2mu
     f_t^{\omega'}\mkern-3mu(x)
\end{equation}
because $(\omega\omega')(a,b,c)=\omega(a,b,c)\mkern2mu\omega'(a,b,c)$ is the multiplication of 3-cocycles.

Since $x_i+x_j+x_k=0$ in~$\mathbb{Z}_2^n$ is equivalent to $x_{i,\sigma}\oplus x_{j,\sigma}\oplus x_{k,\sigma}=0$ for all layers~$\sigma$ we can describe the common $+1$~eigenspace of all triangle operators~$B(p)$ as the subspace stabilized by
\begin{equation}
    Z_{i,
       \sigma}
    Z_{j,
       \sigma}
    Z_{k,
       \sigma}
\end{equation}
for all edges~$i$, $j$ and~$k$ forming a triangle and all layers~$\sigma$.
This subspace is exactly the gauge-invariant subspace of Section~\ref{sec:hamiltonians}.
In order to describe the ground state subspace of the complete Hamiltonian it suffices to add the stabilizers $A_t^\omega(s)$ for all vertices~$s$ and all generators~$t$ of~$\mathbb{Z}_2^n$.
While a vertex operator~$A_t^\omega(s)$ itself may not belong to the Pauli-S group we will find an equivalent stabilizer~$\PauliSVertexStabilizer_t^\omega(s)\in\PauliSGroup$ which coincides with~$A_t^\omega(s)$ on the gauge-invariant subspace.

\subsection{$\mathbb{Z}_2$}
\label{sec:twisted_double_Z2}

The third cohomology group $H^3\bigl(\mathbb{Z}_2,U(1)\bigr)\simeq\mathbb{Z}_2$ is generated by
\begin{equation}
    \omega(a,
           b,
           c)
    =(-1)^{a
           b
           c}.
\end{equation}
It is well known that all twisted quantum double models for the group~$\mathbb{Z}_2$ support Abelian anyons only.

For this~$\omega$ we obtain the phases
\begin{equation}
    f_1^\omega(x)
    =(-1)^{x_1
           x_6+
           x_1
           x_7+
           x_3
           x_4+
           x_3
           x_8+
           x_4
           x_{10}+
           x_6
           x_{11}}\mkern2mu
     (-1)^{x_4+
           x_6+
           x_7+
           x_8}.
\end{equation}
The phases with linear exponent can always be generated by applying~$Z$.
On the gauge-invariant subspace we can also generate all quadratic phases~$(-1)^{x_ix_j}$ by applying suitable powers of~$S$ because the edges~$i$ and~$j$ always belong to some triangle.
Denoting the third edge of the triangle by~$k$ we can indeed get~$(-1)^{x_ix_j}$ from $S_i^3S_j^3S_k$ because $\mathrm{i}^{x_k}=\mathrm{i}^{x_i\oplus x_j}=\mathrm{i}^{x_i+x_j}\mkern2mu(-1)^{x_ix_j}$ holds by Lemma~\ref{lem:parity to high order}.
Hence the operator
\begin{equation}
    X_1
    \cdots
    X_6
    Z_1
    Z_2
    Z_3
    Z_5
    S_7
    S_8
    S_9
    S_{10}^\dagger
    S_{11}^\dagger
    S_{12}
\end{equation}
coincides with~$A_1^\omega(s)$ on the gauge-invariant subspace.
We can recover a more symmetric expression by multiplying with Z-type stabilizers and obtain
\begin{equation}
    \PauliSVertexStabilizer_1^\omega(s)
    =X_1
     \cdots
     X_6
     Z_1
     \cdots
     Z_6
     S_7
     \cdots
     S_{12}.
\end{equation}
This is the same stabilizer as the one in the doubled semion model~\cite{LevinWen} up to conjugation by $S_1\cdots S_6$.
The subspace stabilized by all $\PauliSVertexStabilizer_1^\omega(s)$ and Z-type stabilizers is thus equivalent to the ground state subspace of the doubled semion model up to local unitaries.

Now for a given lattice, we can define $g_j$ with $j\leq t$ to be $\PauliSVertexStabilizer_1^\omega(s)$ on each vertex $s$, and the rest of $g_j$ to be the operator $B(p)$.
One thing needs to be taken care of is when the lattice periodic boundary condition (e.g. torus), $g_j$ will no long be in the standard form as we defined in~\eqref{eq:standard form when no S}, since we have
\begin{equation}
    \prod_{j
           \leq
            t}
    X(\vec{a}_j)
    =\Identity.
\end{equation}
To check that in this case the stabilizer group is still regular, we only need to check the product
\begin{equation}
    \label{eq:product of all g_j on a torus}
    \prod_{j
           \leq
            t}
    g_j
\end{equation}
is a Z-type operator.
We notice that by lemma~\eqref{lem:commutators}, we can exchange $X$ and $S$ in the product~\eqref{eq:product of all g_j on a torus} with the only price being introducing new $Z$ operators into the product.
Thus as long as for each $j$, the $S_j$ operator ($S$ on the $j$th qubit) appears even number of times in the product, we know the product will be a Z-type operator.
And this can be readily checked.
With a straightforward but more involved calculation, we can show that the product~\eqref{eq:product of all g_j on a torus} is satisfied by the gauge-invariant subspace, or in other words, the product can be generated by $\{B(p)\}$.

\subsection{$\mathbb{Z}_2\times\mathbb{Z}_2$}

The third cohomology group $H^3\bigl(\mathbb{Z}_2\times\mathbb{Z}_2,U(1)\bigr)\simeq\mathbb{Z}_2^3$ is generated by
\begin{align}
    \omega_1(a,
             b,
             c)
    & =(-1)^{a_1
             b_1
             c_1}, \\
    \omega_2(a,
             b,
             c)
    & =(-1)^{a_2
             b_2
             c_2}, \\
    \omega_3(a,
             b,
             c)
    & =(-1)^{a_1
             b_2
             c_2}.
\end{align}
It is known that all twisted quantum double models for the group~$\mathbb{Z}_2\times\mathbb{Z}_2$ support Abelian anyons only.

It is not difficult to see that the 3-cocycles~$\omega_1$ and~$\omega_2$ do not lead to anything qualitatively new compared to the case~$\mathbb{Z}_2$.%
\footnote{Indeed, for~$\omega_1$ the phases~$f_{(1,0)}^{\omega_1}$ are confined to layer~1 where we can apply the methods of~\ref{sec:twisted_double_Z2}.
    In contrast, the other generator~$(0,1)$ yields trivial phases only so that $\PauliSVertexStabilizer_{(0,1)}^{\omega_1}=A_{(0,1)}^{\omega_1}$ is an X-type element confined to layer~2.}

The 3-cocycle~$\omega_3$ is much more interesting.
We obtain the phases
\begin{align}
    f_{(1,
        0)}^{\omega_3}(x)
    & =(-1)^{x_{4,
                2}
             x_{10,
                2}+
             x_{6,
                2}
             x_{11,
                2}}, \\
    f_{(0,
        1)}^{\omega_3}(x)
    & =(-1)^{x_{1,
                1}
             x_{6,
                2}+
             x_{3,
                1}
             x_{4,
                2}+
             x_{1,
                2}
             x_{7,
                1}+
             x_{3,
                2}
             x_{8,
                1}}\mkern2mu
       (-1)^{x_{7,
                1}+
             x_{8,
                1}}.
\end{align}
Clearly, the phases associated with~$(1,0)$ are confined to layer~2 and we can apply the methods of~\ref{sec:twisted_double_Z2}.
This results in
\begin{equation}
    \PauliSVertexStabilizer_{(1,
                              0)}^{\omega_3}(s)
    =X_{1,
        1}
     \cdots
     X_{6,
        1}
     S_{4,
        2}^3
     Z_{5,
        2}
     S_{6,
        2}^3
     S_{10,
        2}^3
     S_{11,
        2}^3.
\end{equation}
However, the quadratic phases $(-1)^{x_{i,1}x_{j,2}}$ arising from~$(0,1)$ are of a different kind.
Although all pairs of edges~$i$ and~$j$ continue to belong to some triangle we can no longer exploit the flat connection since the qubits reside on different layers.
Instead we introduce the ancilla qubits
\begin{equation}
    \ket{y_p}
    =\ket{x_{i,
             1}\oplus
          x_{j,
             2}}
\end{equation}
for $(p,i,j)\in\{(1,7,1),(2,8,3),(3,3,4),(4,4,10),(5,6,11),(6,1,6)\}$ and these may be associated with the triangles of the lattice as shown in Figure~\ref{fig:branching}.
Clearly, the above coupling can be enforced by additional Z-type stabilizers.
We will write $\tilde{O}_p$ for an operator~$O$ acting on the ancilla qubit~$y_p$ in the triangle~$p$.
On the gauge-invariant subspace coupled to the ancilla layer we then have
\begin{equation}
    \PauliSVertexStabilizer_{(0,
                              1)}^{\omega_3}(s)
    =X_{1,
        2}
     \cdots
     X_{6,
        2}
     S_{1,
        1}^3
     S_{1,
        2}^3
     S_{3,
        1}^3
     S_{3,
        2}^3
     S_{4,
        2}^3
     S_{6,
        2}^3
     S_{7,
        1}
     S_{8,
        1}
     \tilde{S}_1
     \tilde{S}_2
     \tilde{S}_3
     \tilde{S}_6.
\end{equation}

Similar to~\ref{sec:twisted_double_Z2}, we can also compute the additional diagonal operators when we have a lattice with periodic boundary condition.
Notice that $\PauliSVertexStabilizer_{(1,0)}^{\omega_3}(s)$ ($\PauliSVertexStabilizer_{(0,1)}^{\omega_3}(s)$) commute with each other for any two vertices.
It is then straightforward to check the multiplication of all $\PauliSVertexStabilizer_{(1,0)}^{\omega_3}(s)$ is identity, and $\PauliSVertexStabilizer_{(0,1)}^{\omega_3}(s)$ can be generated by $B(p)$.

\subsection{$\mathbb{Z}_2\times\mathbb{Z}_2\times\mathbb{Z}_2$}
\label{sec:twisted_double_Z2Z2Z2}

The third cohomology group $H^3\bigl(\mathbb{Z}_2\times\mathbb{Z}_2\times\mathbb{Z}_2,U(1)\bigr)\simeq\mathbb{Z}_2^7$ is generated by
\begin{align}
    \omega_1(a,
             b,
             c)
    & =(-1)^{a_1
             b_1
             c_1}, \\
    \omega_2(a,
             b,
             c)
    & =(-1)^{a_2
             b_2
             c_2}, \\
    \omega_3(a,
             b,
             c)
    & =(-1)^{a_3
             b_3
             c_3}, \\
    \omega_4(a,
             b,
             c)
    & =(-1)^{a_1
             b_2
             c_2}, \\
    \omega_5(a,
             b,
             c)
    & =(-1)^{a_1
             b_3
             c_3}, \\
    \omega_6(a,
             b,
             c)
    & =(-1)^{a_2
             b_3
             c_3}, \\
    \omega_7(a,
             b,
             c)
    & =(-1)^{a_1
             b_2
             c_3}.
\end{align}
It turns out that the twisted quantum double models~$\mathrm{D}^\omega(\mathbb{Z}_2^3)$ support non-Abelian anyons if and only if the twist~$\omega$ contains~$\omega_7$ \cite{deWildPropitius}.

Again, the 3-cocycles~$\omega_1$, …, $\omega_6$ lead to situations which qualitatively resemble the cases~$\mathbb{Z}_2$ and~$\mathbb{Z}_2\times\mathbb{Z}_2$.

Now the 3-cocycle~$\omega_7$ leads to truly interesting results.
We obtain the phases
\begin{align}
    f_{(1,
        0,
        0)}^{\omega_7}(x)
    & =(-1)^{x_{4,
                2}
             x_{10,
                3}+
             x_{6,
                2}
             x_{11,
                3}}, \\
    f_{(0,
        1,
        0)}^{\omega_7}(x)
    & =(-1)^{x_{1,
                1}
             x_{6,
                3}+
             x_{3,
                1}
             x_{4,
                3}}, \\
    f_{(0,
        0,
        1)}^{\omega_7}(x)
    & =(-1)^{x_{1,
                2}
             x_{7,
                1}+
             x_{3,
                2}
             x_{8,
                1}}.
\end{align}
Let us introduce the ancilla qubits
\begin{align}
    \ket{y_{p,
            1}}
    & =\ket{x_{i,
               1}\oplus
            x_{j,
               2}}, \\
    \ket{y_{p,
            2}}
    & =\ket{x_{i,
               1}\oplus
            x_{j,
               3}}, \\
    \ket{y_{p,
            3}}
    & =\ket{x_{i,
               2}\oplus
            x_{j,
               3}}
\end{align}
for positions $(p,i,j)\in\{(1,7,1),(2,8,3),(3,3,4),(4,4,10),(5,6,11),(6,1,6)\}$.
This coupling can again be enforced by additional Z-type stabilizers.
We can then write
\begin{align}
    \label{eq:non abelian anyon stabilizer1}
    \PauliSVertexStabilizer_{(1,
                              0,
                              0)}^{\omega_7}(s)
    & =X_{1,
          1}
       \cdots
       X_{6,
          1}
       S_{4,
          2}^3
       S_{10,
          3}^3
       S_{6,
          2}^3
       S_{11,
          3}^3
       \tilde{S}_{4,
                  3}
       \tilde{S}_{5,
                  3}, \\
    \label{eq:non abelian anyon stabilizer2}
    \PauliSVertexStabilizer_{(0,
                              1,
                              0)}^{\omega_7}(s)
    & =X_{1,
          2}
       \cdots
       X_{6,
          2}
       S_{1,
          1}^3
       S_{3,
          1}^3
       S_{4,
          3}^3
       S_{6,
          3}^3
       \tilde{S}_{3,
                  2}
       \tilde{S}_{6,
                  2}, \\
    \label{eq:non abelian anyon stabilizer3}
    \PauliSVertexStabilizer_{(0,
                              0,
                              1)}^{\omega_7}(s)
    & =X_{1,
          3}
       \cdots
       X_{6,
          3}
       S_{1,
          2}^3
       S_{3,
          2}^3
       S_{7,
          1}^3
       S_{8,
          1}^3
       \tilde{S}_{1,
                  1}
       \tilde{S}_{2,
                  1}.
\end{align}

For a given $(j,k,l)\in \{(0,0,1), (0,1,0), (0,0,1)\}$, again $\PauliSVertexStabilizer_{(j,k,l)}^{\omega_7}(s)$ commute with each other for different $s$.
Thus it is easy to compute the product $\prod_s \PauliSVertexStabilizer_{(j,k,l)}^{\omega_7}(s)$ for a lattice with periodic boundary condition.
However, in this case, the product would be some tensor product that contains $S$ operators.
Thus the stabilizer group $G$ for this model on a torus is not a regular XS-stabilizer group, which is different from the previous two models that are based on $\mathbb{Z}_2$ and $\mathbb{Z}_2\times \mathbb{Z}_2$.
However, on a 2D lattice with suitable boundary the stabilizer group~$G$ \emph{is} regular and the unique ground state continues to support non-Abelian anyons since these excitations can be created locally.

\subsection{$\mathbb{Z}_2^n$}

In general, the third cohomology group $H^3\bigl(\mathbb{Z}_2^n,U(1)\bigr)$ is generated by the following types of generators~\cite{deWildPropitius}:
\begin{align}
    \omega_i(a,
             b,
             c)
    & =(-1)^{a_i
             b_i
             c_i}, \\
    \omega_{i
            j}(a,
               b,
               c)
    & =(-1)^{a_i
             b_j
             c_j}, \\
    \omega_{i
            j
            k}(a,
               b,
               c)
    & =(-1)^{a_i
             b_j
             c_k}.
\end{align}
Here $i$, $j$ and~$k$ denote distinct factors (layers) of the direct product group~$\mathbb{Z}_2^n$.
We have shown above how the phases $f_t^\omega$ for each such generator~$\omega$ can be expressed within the XS-stabilizer formalism by coupling ancilla qubits to the original ones as necessary.

This clearly extends to arbitrary elements of the third cohomology group.
Suppose we want to obtain the phases associated with the 3-cocycle~$\omega\omega'$ where $\omega$ and~$\omega'$ are any of the above generators.
From~\eqref{eq:phases_product_cocycle} we see that we can construct these phases independently for~$\omega$ and $\omega'$.
This shows that we can describe the ground state subspaces of arbitrary twisted quantum double models~$\mathrm{D}^\omega(\mathbb{Z}_2^n)$ with our XS-stabilizer formalism.

\begin{raggedright}
    \printbibliography
\end{raggedright}

\end{document}